\documentclass[journal,twosided,web]{ieeecolor}
\usepackage{generic}
\usepackage{cite}
\usepackage{amsmath,amssymb,amsfonts, cite}
\usepackage{algorithmic}
\usepackage{graphicx}
\usepackage{textcomp}
\usepackage{dsfont}
\usepackage{color}
\usepackage{epstopdf}        
\usepackage{enumerate}
\usepackage{soul}
\usepackage{lscape}
\usepackage{multicol}
\usepackage{multirow}
\usepackage{calligra}
\usepackage{booktabs}
\usepackage{mathtools}
\usepackage{framed} 
\usepackage{picins}
\usepackage{empheq}
\usepackage{afterpage}
\usepackage{color}
\newcommand{\tcb}{\textcolor{black}}
\usepackage{graphics} % for pdf, bitmapped graphics files
\usepackage{epsfig} % for postscript graphics files
\usepackage{times} % assumes new font selection scheme installed
\usepackage{amsmath} % assumes amsmath package installed
\usepackage{amssymb}  % assumes amsmath package installed
\usepackage{amsfonts}  % assumes amsmath package installed
\usepackage{subfig}
\usepackage{epstopdf}
\usepackage{enumerate}
\usepackage{graphicx} % grÂficos
\usepackage{xcolor}
\usepackage{comment}

\DeclareMathOperator*{\argmin}{arg\,min}
\DeclarePairedDelimiter\sg{\lceil}{\rfloor}
\newtheorem{thm}{Theorem}

\newtheorem{lemma}{Lemma}
\newtheorem{cor}{Corollary}
\newtheorem{definition}{Definition}
\newtheorem{assumption}{Assumption}
\newtheorem{remark}{Remark}

\newcommand\bj{\mathbf{J}}
\newcommand\re{\mathbb{R}}
\newcommand\fx{\mathcal{F}_{\xi_1, \xi_2}}
\newcommand\hx{\hat{x}}
\newcommand\vt{\varphi(\theta)}
\newcommand*{\QEDB}{\hfill\ensuremath{\square}}%

\def\BibTeX{{\rm B\kern-.05em{\sc i\kern-.025em b}\kern-.08em
    T\kern-.1667em\lower.7ex\hbox{E}\kern-.125emX}}
\begin{document}
\title{Fixed-Time Input-to-State Stability for Singularly Perturbed Systems via Composite Lyapunov Functions\\\vspace{0.2cm} (Extended Manuscript) }

\author{Michael Tang, Miroslav Krsti\'c, Jorge I. Poveda
%\thanks{This paragraph of the first footnote will contain the date on 
%which you submitted your paper for review. It will also contain support 
%information, including sponsor and financial support acknowledgment. For 
%example, ``This work was supported in part by the U.S. Department of 
%Commerce under Grant BS123456.'' }
\thanks{M. Tang and J. I. Poveda are with the Dep. of Electrical and Computer Engineering, University of California, San Diego, La Jolla, CA, USA.}
\thanks{ M. Krsti\'c is with the Dep. of Mechanical and Aerospace Engineering, University of California, San Diego, La Jolla, CA, USA.}
\thanks{This work was supported in part by NSF grants ECCS CAREER 2305756, CMMI 2228791, and AFOSR FA9550-22-1-0211.}
\vspace{-0.6cm}}
\maketitle

\begin{abstract}
We study singularly perturbed systems that exhibit input-to-state stability (ISS) with fixed-time properties in the presence of bounded disturbances. In these systems, solutions converge to the origin within a time frame independent of initial conditions when undisturbed, and to a vicinity of the origin when subjected to bounded disturbances. First, we extend the traditional composite Lyapunov method, commonly applied in singular perturbation theory to analyze asymptotic stability, to include fixed-time ISS. We demonstrate that if both the reduced system and the boundary layer system exhibit fixed-time ISS, and if certain interconnection conditions are met, the entire multi-time scale system retains this fixed-time ISS characteristic, provided the separation of time scales is sufficiently pronounced. Next, we illustrate our findings via analytical and numerical examples, including a novel application in fixed-time feedback optimization for dynamic plants with slowly varying cost functions. 
\end{abstract}

%Despite this relevance, the study of stable multi-agent dynamical systems able to model deception in decision-making algorithms remains largely unexplored.
%Traditionally, equilibrium-seeking problems in competitive settings, such as non-cooperative games, have been studied under the assumption of \emph{symmetric information} among all decision-makers, where each agent has equal access to information for decision-making. However, in realistic competitive settings, some agents often gain privileged information about others' decision rules, breaking this symmetry. Leveraging asymmetric information through systematic \emph{deception} has received significant attention in game theory, cybersecurity, and multi-agent coordination. 
%%%%%%%%%%%%%%%%%%
% \begin{IEEEkeywords}
% Learning in Games, Nash equilibria, Non-cooperative games, Equilibrium Seeking.
% \end{IEEEkeywords}
%%%%%%%%%%%%%%%%%%%
%
\vspace{-0.2cm}
\section{Introduction}
Many complex dynamical systems can be decomposed into components operating on different time scales. Such multi-time-scale systems are common across a diverse array of engineering applications, including aerospace \cite{narang2014nonlinear}, chemical processing \cite{261484}, smart grids \cite{5467283}, and biological and evolutionary systems \cite{del2015biomolecular}, among others. %These types of systems can often be represented in the following form
%
% \begin{subequations}\label{spsys}
%     \begin{align}
%         \dot{x}&=f(x,z,t,\varepsilon)\\
%         \varepsilon\dot{z}&=g(x,z,t,\varepsilon)
%     \end{align}
% \end{subequations}
%
%When $\varepsilon$ is small, it creates a time scale separation between states $x$ and $z$, which are respectively referred to as the ``slow" and ``fast" states.
Such systems can be studied using techniques from singular perturbation theory \cite{vasil1978singular,KokotovicSPBook, khalil}, where the methods typically focus on decomposing the system into lower order subsystems and studying the subsystems to draw conclusions on the behavior of the overall interconnection.
\par
%

%\subsection{Literature Review}
%
One simple yet powerful technique for analyzing singularly perturbed systems is the \emph{composite Lyapunov method}, first introduced by Khalil in \cite{1103586}. This method is particularly effective for assessing system stability, as it leverages the stability of the subsystems to construct a Lyapunov function for the overall system that is valid under sufficiently large time-scale separation. As demonstrated in \cite{1103586}, if the reduced and boundary layer systems admit certain quadratic-type Lyapunov functions, then under additional interconnection conditions, the interconnected system can be shown to be asymptotically stable, provided there is sufficient time-scale separation between the two subsystems. Moreover, if the Lyapunov functions satisfy specific quadratic bounds, these results can be extended to exponential stability. The composite Lyapunov method has found profound applications in control synthesis across various engineering domains, including power systems \cite{9309064}, biological systems \cite{9029355}, aerospace \cite{naidu2001singular}, etc. {While the composite Lyapunov method has been widely used to establish asymptotic and exponential stability, it has found limited applications in settings which require stronger notions of stability. For example, in \cite{9992641}, the authors leverage composite Lyapunov functions to establish finite-time stability for a special class of singularly perturbed homogeneous systems. However, to the best of the authors' knowledge, there does not exist a framework that studies fixed-time stability for a general class of singularly perturbed systems via composite Lyapunov techniques.}

\par

Fixed-time stability, {which has been} popularized through the introduction of Lyapunov conditions in \cite{Fixed_timeTAC}, has garnered significant attention due to its ability to address challenges in control \cite{polyakov2023finite}, optimization \cite{fixed_time}, and learning \cite{nonsmoothesc,poveda2022fixed}. Fixed-time stability ensures convergence to an equilibrium within a fixed time, regardless of the system's initial conditions. While the Lyapunov conditions simplify the verification of fixed-time stability in continuous-time dynamical systems, techniques for analyzing fixed-time stability in interconnected systems remain limited. Current state-of-the-art methods are primarily applicable to systems with specific homogeneity properties \cite{mendoza2023stability}, restrictive structural requirements \cite{lei2022event}, or those {satisfying certain small-gain conditions \cite{10886634}}. {Moreover,} in many practical applications, the systems of interest are also influenced by external disturbances or exogenous inputs. In particular, in this paper we study systems of the form
\begin{subequations}\label{sys}
    \begin{align}
        \dot{x}&=f(x,z,u)\label{sysx}\\
        \varepsilon\dot{z}&=g(x,z,u)\label{sysz},
    \end{align}
\end{subequations}
where $\varepsilon>0$ is a small parameter that induces a time scale separation between the dynamics of $x$ and the dynamics of $z$, and $u$ is an exogenous input. One of the key tools for analyzing such systems is input-to-state stability (ISS), introduced by Sontag in \cite{28018} and later characterized through an equivalent Lyapunov framework in \cite{SONTAG1995351}. In simple terms, an ISS system converges asymptotically to a bounded set for any bounded input signal. A fixed-time counterpart, known as fixed-time ISS (FxT ISS), was introduced in \cite{LOPEZRAMIREZ2020104775} to address fixed-time stability in systems subject to disturbances. While standard ISS properties have been extensively studied in the context of singularly perturbed systems \cite{544001,TeelNesicTAC}, it remains an open question whether analogous results can be developed for fixed-time stability and whether the composite Lyapunov method introduced in \cite{1103586} can be extended to analyze FxT ISS.
\par

In this paper, we address the above questions by extending the popular composite Lyapunov method \cite{1103586} to study FxT ISS in singularly perturbed systems with inputs. Specifically, under the assumptions that the respective ``reduced'' and ``boundary-layer'' dynamics of \eqref{sys} are fixed-time stable on their own, we first derive interconnection conditions under which \eqref{sys} renders the origin FxT ISS. These interconnection conditions parallel those developed in \cite{1103586} for asymptotic and exponential stability. {However, unlike the results of \cite{1103586}, absence of Lipschitz continuity at the origin results in a more complex paradigm for verifying our proposed interconnection conditions, which involve powers of the fixed-time Lyapunov functions of the reduced and boundary-layer dynamics}. Subsequently, we employ our analytical tools to design and analyze a novel class of algorithms that achieve fixed-time feedback optimization under strongly convex cost functions and slowly varying inputs---a class of problems widely studied in the literature~\cite{colombino2019online,hauswirth2020timescale,bianchin2022online}, but which, to the best of our knowledge, has not previously been addressed using algorithms with fixed-time stability properties. Finally, we present numerical and analytical examples to illustrate the key assumptions and the main stability results of the paper.

\par
%

% \vspace{-0.3cm}
% \subsection*{Additional Contributions with respect to \cite{10644358}}
%
Earlier, preliminary results from this work were previously published in the proceedings of the American Control Conference \cite{10644358}. However, the results of \cite{10644358} focused exclusively on systems without inputs, addressing only fixed-time stability, and provided only proof sketches. In contrast, this paper extends the analysis to singularly perturbed systems \emph{with inputs}, advancing the FxT stability results of \cite{10644358} to encompass fixed-time input-to-state stability (FxT ISS). Furthermore, in this paper we present a comprehensive stability analysis, complete proofs, and novel analytical and numerical examples to illustrate the main findings. This includes an application to feedback optimization under slowly varying inputs in general fixed-time stable plants. 

\vspace{0.2cm}
The rest of this paper is organized as follows. Section \ref{sec_prelim} presents the notation and preliminary results. Section \ref{sec_prob} discusses the problem of interest, the main stability results, and showcases the applicability of our results in an illustrative example. Section \ref{sec_ex} studies fixed-time feedback optimization, and Section \ref{sec_conc} ends with the conclusions.
\section{Preliminaries}\label{sec_prelim}
\subsection{Notation}
We use $\mathbb{R}_{\ge 0}$ to denote the set of nonnegative real numbers. We let $\text{sgn}:\mathbb{R}\to \{-1,{0,} 1\}$ denote the sign function, i.e $\text{sgn}(x)=1$ if $x> 0$, $\text{sgn}(x)=-1$ if $x<0$ and $\text{sgn}(0)=0$. For a continuous function $\alpha:\mathbb{R}_{\ge 0}\to\mathbb{R}_{\ge 0}$, we say $\alpha\in\mathcal{K}$ (i.e, $\alpha$ is of class $\mathcal{K}$) if $\alpha(0)=0$ and $\alpha$ is strictly increasing. If $\alpha\in\mathcal{K}$ also satisfies $\lim_{s\to\infty}\alpha(s)=\infty$, we say $\alpha\in\mathcal{K}_\infty$. Given a continuous function $\beta:\mathbb{R}_{\ge 0}\times\mathbb{R}_{\ge 0}\to\mathbb{R}_{\ge 0}$, we say $\beta\in\mathcal{KL}$ if for each $t\ge 0,$ $\beta(\cdot, t)\in\mathcal{K}$ and for each $r\ge 0$, $\beta(r, \cdot)$ is non-increasing and asymptotically goes to 0. Furthermore, $\beta\in\mathcal{GKL}$ if $\beta(\cdot, 0)\in\mathcal{K}$ and for each fixed $r\ge 0$, $\beta(r, \cdot)$ is continuous, non-increasing and there exists {a function $T:\mathbb{R}_{\ge 0}\to\mathbb{R}_{\ge 0}$ such that $\beta(r, t)=0$ for all $t\ge T(r)$}. The mapping $T$ is called {a \emph{settling time function}, which, in general, is not unique}. Given a measurable function $u:\mathbb{R}_{\ge 0}\to\mathbb{R}^p$ we denote $|u|_\infty=\text{ess}\sup_{t\ge 0}|u(t)|$. We use $\mathcal{L}_\infty^p$ to denote the set of measurable functions $u:\mathbb{R}_{\ge 0}\to\mathbb{R}^p$ satisfying $|u|_\infty<\infty$. Given a differentiable function $f:\mathbb{R}^n\to\mathbb{R}^m$, we use $\bj_f(x)\in\mathbb{R}^{m\times n}$ to denote the Jacobian of $f$ evaluated at $x\in\mathbb{R}^n$. If $m=1$, we use $\nabla f(x)=\bj_f(x)^\top$. If $\bj_f(x)$ is continuous, we say $f$ is $\mathcal{C}^1$.
\subsection{Auxiliary Results}
We first present some lemmas that will be instrumental for our results. %Due to space limitations, the proofs can be found in the extended manuscript \cite{tang2024fixed}. 
\begin{lemma}\label{amgm}
    Given $x, y\ge 0$ and $p_1, p_2>0$, the following inequality holds for all {$c>0$}:
    \begin{equation*}
        |x|^{p_1}|y|^{p_2}\le {c}|x|^{p_1+p_2}+{c^{-\frac{p_1}{p_2}}}|y|^{p_1+p_2}.
    \end{equation*}
\end{lemma}\QEDB
% \begin{lemma}\label{sandwich}
%     Given $x\ge 0$ and $a\le b\le c$, then $x^b\le x^a+x^c$.
% \end{lemma}\QEDB
\begin{lemma}\label{bounds}
    For $\xi_1\in(0,1)$ and $\xi_2<0$, the following inequalities hold for all $x,y\in\mathbb{R}^n$:
    \begin{subequations}
        \begin{align}
            \left\lvert \dfrac{x}{|x|^{\xi_1}}-\dfrac{y}{|y|^{\xi_1}} \right\rvert&\le 2^{\xi_1}|x-y|^{1-\xi_1}\label{bdxi1}\\
            \left\lvert \dfrac{x}{|x|^{\xi_2}}-\dfrac{y}{|y|^{\xi_2}} \right\rvert&\le K|y-x|\left(|x|^{-\xi_2}+|y-x|^{-\xi_2}\right)\label{bdxi2}.
        \end{align}
    \end{subequations}
    where $K:=1+\max\{1,-\xi_2 2^{-\xi_2-1}\}$.\QEDB
\end{lemma}
% \begin{lemma}\label{misk}
%     Given $\beta\in\mathcal{KL}$ and some set $K\subset\mathbb{R}_{\ge 0}$, define $M_\beta(r)=\sup_{t\in K}\beta(r, t)$. Then $M_\beta\in\mathcal{K}$.
% \end{lemma}
% \begin{proof}
%     Since $\beta(\cdot, t)$
% \end{proof}
% \begin{lemma}
%     Suppose $\beta\in\mathcal{GKL}$ and the settling time function satisfies $T_0=0$. Then there exists $\tilde{\beta}\in\mathcal{GKL}\cap\mathcal{KL}$ such that $\beta(r,t)\le\tilde{\beta}(r,t)$ for all $r, t\ge 0$
% \end{lemma}
% \begin{proof}
%     Majorize
% \end{proof}
\vspace{0.1cm}
\subsection{Fixed-Time Input-to-State Stability}
Consider a nonlinear dynamical system of the form
\begin{equation}
    \dot{x}=f(x,u),\quad x(0)=x_0\label{xusys},
\end{equation}
where $x\in\mathbb{R}^n$ is the state and $u\in\mathcal{L}_\infty^p$ is an input signal. We assume the vector field $f$ is continuous and satisfies $f(0,0)=0$. 
% Given an initial condition $x_0=x_0$ and input signal $u(t)$ we let $\Phi(t, x_0, u)$ denote the solution of \eqref{xusys} for $t\ge 0$. Since the solution might not be unique, we assume a certain property holds for all solutions generated from a given initial condition if the desired property is satisfied for that initial condition. 
We will state some definitions from \cite{LOPEZRAMIREZ2020104775} that will be particularly relevant for our work.
\begin{definition}\label{fxtissdef}
     System \eqref{xusys} is said to be \emph{fixed-time input-to-state stable (FxT ISS)} if for each $x_0\in\mathbb{R}^n$ and $u\in\mathcal{L}_\infty^p$, every solution $x(t)$ of \eqref{xusys} exists for $t\ge 0$ and satisfies
     \begin{equation}\label{fxtISSbounddef}
         |x(t)|\le \beta(|x_0|, t)+\varrho(|u|_\infty),
     \end{equation}
     where $\beta\in\mathcal{GKL}$, $\varrho\in\mathcal{K}$ and {there exists a settling time function $T$ of $\beta$ that is continuous and uniformly bounded, with $T(0)=0$}.
\end{definition}
\begin{definition}\label{fxtissv}
    A $\mathcal{C}^1$ function $V:\mathbb{R}^n\to\mathbb{R}_{\ge 0}$ is called a \emph{FxT ISS Lyapunov function} for \eqref{xusys} if there exists $\alpha_1, \alpha_2\in\mathcal{K}_\infty$ such that
    \begin{equation}\label{fxtiss_sw}
        \alpha_1(|x|)\le V(x)\le \alpha_2(|x|),
    \end{equation}
    and the following holds
    \begin{equation}
        \dfrac{\partial V}{\partial x}f(x, u)\le -k_1 V^{a_1}{(x)}-k_2 V^{a_2}{(x)}+\rho(|u|)\label{fxtiss_imp},
    \end{equation}
    for some $\rho\in\mathcal{K}_\infty$, $k_1, k_2>0$, $a_1\in (0,1)$ and $a_2>1$.
\end{definition}

\vspace{0.1cm}
{\begin{remark}
    It can be verified that the ``dissipation" formulation we use in \eqref{fxtiss_imp} implies relation (10) in \cite{LOPEZRAMIREZ2020104775}, which implies FxT ISS for \eqref{xusys} via \cite[Thm 4]{LOPEZRAMIREZ2020104775}. Indeed, if \eqref{fxtiss_imp} holds and we fix some $0<\tilde{\varepsilon}<\min_i k_i$, we have
    \begin{equation*}
       |x|>\chi(|u|)\Rightarrow \dfrac{\partial V}{\partial x}f(x, u)\le -\tilde{k}_1 V^{a_1}{(x)}-\tilde{k}_2 V^{a_2}{(x)}, 
    \end{equation*}
    where $\chi(\cdot)=\tilde{\chi}^{-1}\circ\rho(\cdot)$, \tcb{$\tilde{\chi}(\cdot)=\tilde{\varepsilon}\alpha_1^{a_1}(\cdot)+\tilde{\varepsilon}\alpha_1^{a_2}(\cdot)$}, and $\tilde{k}_i=k_i-\tilde{\varepsilon}$. By taking $\tilde{\varepsilon}\to 0$, it can be observed via \cite[Corollary 2]{LOPEZRAMIREZ2020104775} that systems that admit a FxT ISS Lyapunov function that satisfy \eqref{fxtiss_imp} also admit a settling time function satisfying the following bound
    \begin{equation}\label{settime}
        T(x_0)\le \frac{1}{k_1(1-a_1)}+\frac{1}{k_2(a_2-1)},
    \end{equation}
    for all $x_0\in\mathbb{R}^n$.
\end{remark}}
\section{Main Results}\label{sec_prob}
We consider singularly perturbed systems of the form \eqref{sys}, with states $x \in \mathbb{R}^{n}$ and $z \in \mathbb{R}^{m}$, input $u \in \mathbb{R}^{p}$, dynamics satisfying {$f(0,z^*,0) = g(0,z^*,0) = 0$ for some $z^*\in\re^m$, and a small parameter $\varepsilon > 0$ that induces a time scale separation between the dynamics of $x$ and $z$. Our main objective is to exploit the stability properties of the so-called \emph{reduced system} and \emph{boundary-layer system} associated with \eqref{sys}, as defined below, in order to derive Lyapunov-based sufficient conditions that ensure fixed-time input-to-state stability (FxT ISS) of \eqref{sys}, provided the time-scale separation is sufficiently large.}
\subsection{Assumptions}
To study system \eqref{sys} in the context of singular perturbations, we make the following standard assumption on \eqref{sysz}:
\begin{assumption}\label{aqss}
    There exists a $\mathcal{C}^1$ mapping $h:\mathbb{R}^{n}\to \mathbb{R}^{m}$ such that $g(x, z, 0)=0$ if and only if $z=h(x)$.\QEDB
\end{assumption}

\vspace{0.1cm}
The map $h$ is usually referred to as the \emph{quasi-steady state mapping} \cite{khalil} for the disturbance-free system \eqref{sysz}. By using this mapping, we can define the so-called \emph{reduced system} from \eqref{sysx}:
\begin{equation}\label{rsys}
    \dot{x}=f(x, h(x),u).
\end{equation}
{Using the change of coordinates $y=z-h(x)$, we obtain the following error dynamics:}
\begin{subequations}\label{xysys}
    \begin{align}
    \dot{x}&=f(x, y+h(x),u)\\
        \dot{y}&=\dfrac{1}{\varepsilon}g(x,y+h(x),u)-\dfrac{\partial h}{\partial x}f(x,y+h(x),u)\label{dy}.
    \end{align}
    \end{subequations}
%\vspace{-0.3cm}
%{\begin{remark}In the classic singular perturbation literature, if the fast dynamics are unaffected by the input, the state $y$ is usually referred to as an ``error" state. If the fast dynamics depend on the input, the state $y$ can be viewed more generally as a change of variables made to simplify the stability analysis.
%\end{remark}} 
System \eqref{dy} is studied in the time scale $\tau=t/\varepsilon$ and taking $\varepsilon\to 0^+$ to obtain the \emph{boundary layer system}:
\begin{equation}\label{blsys}
    \dfrac{dy}{d\tau}=g(x, y+h(x),u),
\end{equation}
{where $x\in\mathbb{R}^{n}$ is considered fixed.} Since our goal is to study FxT ISS of \eqref{sys}, we make the following FxT ISS Lyapunov-based assumptions on the lower order systems \eqref{rsys} and \eqref{blsys}: 
\begin{assumption}\label{assump_rsys}
     There exists a $\mathcal{C}^1$ function $V:\mathbb{R}^{n}\to\mathbb{R}_{\ge 0}$ and $\alpha_1, \alpha_2, \rho_R\in\mathcal{K}_\infty$ such that
    \begin{align*}
        \alpha_1(|x|)\le V(x)\le \alpha_2(|x|),
    \end{align*}
    and
    \begin{equation*}
        \dfrac{\partial V}{\partial x}f(x, h(x),u)\le -k_1 V^{a_1}{(x)}-k_2 V^{a_2}{(x)}+\rho_R(|u|),
    \end{equation*}
    where $k_1, k_2>0$, $a_1\in (0,1)$ and $a_2>1$.\QEDB
\end{assumption}
\vspace{0.1cm}
\begin{assumption}\label{assump_blsys}
    There exists a $\mathcal{C}^1$ function $W:\mathbb{R}^{n}\times\mathbb{R}^{m}\to\mathbb{R}_{\ge 0}$ and $\tilde{\alpha}_1, \tilde{\alpha}_2, \rho_B\in\mathcal{K}_\infty$ such that
    \begin{align*}
        \tilde{\alpha}_1(|y|)\le W(x,y)\le \tilde{\alpha}_2(|y|),
    \end{align*}
    and
    \begin{small}
    \begin{equation*}
        \dfrac{\partial W}{\partial y}g(x, y+h(x),u)\le -\kappa_1 W^{b_1}{(x,y)}-\kappa_2 W^{b_2}{(x,y)}+\rho_B(|u|),
    \end{equation*}
    \end{small}
    where $\kappa_1, \kappa_2>0$, $b_1\in (0,1)$ and $b_2>1$.\QEDB
\end{assumption}

\vspace{0.1cm}
{Assumptions \ref{assump_rsys} and \ref{assump_blsys} state, respectively, that the reduced system \eqref{rsys}
and the boundary-layer system \eqref{blsys} admit individual FxT ISS Lyapunov functions and are, hence, FxT ISS.} These assumptions mirror the classical Lyapunov-based conditions commonly used in the literature to study the asymptotic stability of singularly perturbed systems. However, as shown in \cite{1103586} and \cite{544001}, the mere stability or ISS properties of the individual reduced and boundary-layer systems are usually insufficient to guarantee that \eqref{sys} is also stable or ISS, and additional interconnection conditions need to be examined. %While for Lipschitz systems such conditions can be readily studied, for fixed-time stable systems the dynamics of interest are never Lipschitz continuous, which makes the analysis challenging.
\subsection{Analysis}
{To assess the FxT stability of system \eqref{sys} using a Lyapunov-based approach, we consider the following Lyapunov function candidate
    \begin{equation}
        \Psi_\zeta(x,y)=\zeta V(x)+(1-\zeta)W(x,y),\quad \zeta\in (0,1),\label{compl}
    \end{equation}
    where $V$ and $W$ are from Assumptions \ref{assump_rsys} and \ref{assump_blsys} respectively. For historical reasons, we refer to \eqref{compl} as a \emph{composite Lyapunov function} \cite{1101342,1104064}. Evaluating the Lie derivative of \eqref{compl} along the trajectories of \eqref{xysys} results in
    \begin{align}
        \dot{\Psi}_\zeta&=\zeta\left(\frac{\partial V}{\partial x}f(x, h(x),u)+ I_1(x,y,u)\right)\notag\\&~~~+(1-\zeta)\left(\frac{1}{\varepsilon}\frac{\partial W}{\partial y}g(x, y+h(x),u)+I_2(x,y,u)\right),\label{compld}
    \end{align}
where the \emph{interconnection terms}, $I_1$ and $I_2$, are given by
\begin{subequations}\label{int}
    \begin{align}
        I_1{(x,y,u)}&=\dfrac{\partial V}{\partial x}\left(f(x,y+h(x),u)-f(x,h(x),u)\right)\\
        I_2{(x,y,u)}&=\left(\dfrac{\partial W}{\partial x}-\dfrac{\partial W}{\partial y}\dfrac{\partial h}{\partial x}\right)f(x,y+h(x),u).
    \end{align}
\end{subequations}}
Ideally, we aim to derive suitable bounds on these terms that would allow us to conclude that \eqref{sys} is fixed-time ISS. Before we do so, we define the following terms:
\begin{subequations}\label{definitionsproof}
\begin{align}
        \tilde{V}(x)&:=V(x)^{\frac{a_1}{2}}+ V(x)^{\frac{a_2}{2}}\\
        \tilde{W}(x, y)&:=W(x,y)^{\frac{b_1}{2}}+ W(x,y)^{\frac{b_2}{2}}\\
        \underline{k}&:=\min_i k_i,~~\underline{\kappa}:=\min_i \kappa_i,
\end{align}
\end{subequations}
where $V, W, a_i,b_i,k_i,\kappa_i$ come from Assumptions \ref{assump_rsys}-\ref{assump_blsys}. 

With these definitions at hand, we can now state the first main result of the paper.
\begin{thm}\label{issthm}
    Consider system \eqref{xysys}, and suppose that Assumptions \ref{aqss}-\ref{assump_blsys} hold. Furthermore, suppose there exists $\nu_1,\nu_2, \omega_1,\omega_2\in\mathbb{R}$ and $\rho_1, \rho_2\in\mathcal{K}_\infty$ such that the interconnection terms in \eqref{int} satisfy 
    \begin{subequations}\label{intcond}
    \begin{align}
        I_1{(x,y,u)}&\le \nu_1 \tilde{V}^2{(x)} +\omega_1 \tilde{W}^2{(x,y)}+\rho_1(|u|)\label{intu1}\\
        I_2{(x,y,u)}&\le \nu_2 \tilde{V}^2{(x)} +\omega_2 \tilde{W}^2{(x,y)}+\rho_2(|u|)\label{intu2}\\
        \nu_1&< \frac12 \underline{k}, \quad\text{or}\quad \nu_2<0\label{vcond}.
    \end{align}
\end{subequations}
Then, there exists $\varepsilon^*>0$ such that, for each $\varepsilon\in(0, \varepsilon^*)$, the system \eqref{xysys} is FxT ISS. \QEDB
\end{thm}

\vspace{0.1cm}
\begin{proof}
    Continuing from \eqref{compld}, we can use Assumptions \ref{assump_rsys} and \ref{assump_blsys} to obtain:
    \begin{align*}
        \dot{\Psi}_\zeta&\le -\zeta\underline{k}(V^{a_1}{(x)}+V^{a_2}{(x)})+\zeta\rho_R(|u|)+\zeta I_1+(1-\zeta)I_2\\&~~~-\frac{1-\zeta}{\varepsilon}\underline{\kappa}(W^{b_1}{(x,y)}+W^{b_2}{(x,y)})+\frac{1-\zeta}{\varepsilon}\rho_B(|u|)\\
        &\le -\dfrac{\zeta}{2} \underline{k}\tilde{V}^2{(x)}-\dfrac{1-\zeta}{2\varepsilon}\underline{\kappa}\tilde{W}^2{(x,y)}+\zeta I_1+(1-\zeta)I_2\\&~~~+\zeta \rho_R(|u|)+\frac{1-\zeta}{\varepsilon}\rho_B(|u|),
    \end{align*}
    where we used Jensen's inequality to obtain $2(V^{a_1}+V^{a_2})\geq\tilde{V}^2$ and $2(W^{b_1}+W^{b_2})\geq\tilde{W}^2$, and for simplicity, we omit the arguments of $I_1$ and $I_2$.
    With conditions \eqref{intu1} and \eqref{intu2}, we have:
    \begin{align*}
        \dot{\Psi}_\zeta
        &\le-\nu(\zeta)\tilde{V}^2{(x)}-\omega_\varepsilon(\zeta)\tilde{W}^2{(x,y)}+\rho_{\varepsilon, \zeta}(|u|),
    \end{align*}
    where
    \begin{subequations}
        \begin{align*}
\nu(\zeta)&=\zeta\left(\dfrac12\underline{k}-\nu_1\right)-(1-\zeta)\nu_2\\
\omega_\varepsilon(\zeta)&=\dfrac{1-\zeta}{2\varepsilon}\underline{\kappa}-\zeta\omega_1-(1-\zeta)\omega_2\\
\rho_{\varepsilon, \zeta}(s)&=\zeta\left(\rho_R(s)+\rho_1(s)\right)+(1-\zeta)\left(\frac{\rho_B(s)}{\varepsilon}+\rho_2(s)\right),
        \end{align*}
    \end{subequations}
    and we clearly have $\rho_{\varepsilon, \zeta}\in\mathcal{K}_\infty$ for $\zeta\in(0,1)$ and $\varepsilon>0$. By \eqref{vcond} we can find $\zeta^*\in(0,1)$ such that $\nu^*:=\nu(\zeta^*)>0$, and then we can find $\varepsilon^*>0$ such that $\omega_\varepsilon(\zeta^*)>\nu^*$ whenever $\varepsilon\in(0,\varepsilon^*)$. For these $\varepsilon\in(0, \varepsilon^*)$, we obtain:
    \begin{align*}
        \dot{\Psi}_{\zeta^*}&\le -\nu^*\left(\tilde{V}^2{(x)}+\tilde{W}^2{(x,y)}\right)+\rho_{\varepsilon, \zeta^*}(|u|)\\&\le
        -\nu^*\left(V^{a_1}{(x)}+V^{a_2}{(x)}+W^{b_1}{(x,y)}+W^{b_2}{(x,y)}\right)\\&~~~+\rho_{\varepsilon, \zeta^*}(|u|)\\
        &=-\dfrac{\nu^*}{2}\bigl((V^{a_1}{(x)}+V^{a_2}{(x)})+(V^{a_1}{(x)}+V^{a_2}{(x)})\\&~~~+(W^{b_1}{(x,y)}+W^{b_2}{(x,y)})+(W^{b_1}{(x,y)}+W^{b_2}{(x,y)})\bigl)\\&~~~+\rho_{\varepsilon, \zeta^*}(|u|),
    \end{align*}\noindent 
    where the second inequality follows by the fact that $(a+b)^2\geq a^2+b^2$ for all $a,b\geq0$. We pick $\gamma_1\in \left[\max\left\{a_1, b_1\right\},1\right)$ and $\gamma_2\in\left(1,\min\left\{a_2, b_2\right\} \right]$ to obtain the following inequality via Lemmas \ref{jensenlemma}-\ref{lem_sandw} in the Appendix:
    \begin{align*}
        \dot{\Psi}_{\zeta^*}&\le -\dfrac{\nu^*}{2}\left(V^{\gamma_1}{(x)}+V^{\gamma_2}{(x)}+W^{\gamma_1}{(x,y)}+W^{\gamma_2}{(x,y)}\right)\\&~~~+\rho_{\varepsilon, \zeta^*}(|u|)\\
        &\le -\dfrac{\nu^*}{2}\left((V{(x)}+W{(x,y)})^{\gamma_1}+2^{1-\gamma_2}(V{(x)}+W{(x,y)})^{\gamma_2}\right)\\&~~~+\rho_{\varepsilon, \zeta^*}(|u|)\\
        &\le -\dfrac{\nu^*}{2}\left(\Psi_{\zeta^*}^{\gamma_1}+2^{1-\gamma_2}\Psi_{\zeta^*}^{\gamma_2}\right)+\rho_{\varepsilon, \zeta^*}(|u|).
    \end{align*}
    Hence, \eqref{xysys} is FxT ISS for $\varepsilon\in(0, \varepsilon^*)$.
\end{proof}
\vspace{0.1cm}
% \begin{remark}
%     The case where the right hand sides of \eqref{intu1} and \eqref{intu2} also contain a $c_i\tilde{V}\tilde{W}$ term can be reduced to \eqref{intu1} and \eqref{intu2} by applying Lemma \ref{amgm}. \QEDB 
% \end{remark}
{ We would like to note that the ``or" condition in \eqref{vcond} is non-exclusive, \tcb{i.e.,} it is acceptable for both conditions to be satisfied. In this case, any choice of the weight $\zeta\in (0,1)$ will result in a valid composite Lyapunov function candidate.
\begin{remark}
    The functions $\tilde{V}(x)$ and $\tilde{W}(x,y)$ are highly analogous to the functions $\psi_1(y)$ and $\psi_2(y)$, respectively, from \cite[Chapter 11.5]{khalil}. For the asymptotic stability (resp. FxT ISS) result from \cite[Theorem 11.3]{khalil} (resp. Theorem \ref{issthm} in this paper), the Lie derivatives of $V$ and $W$ along the reduced and boundary layer dynamics are assumed to be upper bounded by negative multiples of $\psi_1^2(x)$ and $\psi_2^2(y)$ (resp. $\tilde{V}^2(x)$ and $\tilde{W}^2(x,y)$), respectively. Hence, we can observe that the structure of the interconnection conditions in our paper are actually more forgiving than those from \cite[Chapter 11.5]{khalil} in the sense that we allow for extra $\tilde{V}^2(x)$ and $\tilde{W}^2(x,y)$ terms. This extra degree of freedom is particularly useful since it allows us to apply our results to a variety of interesting \tcb{interconnected} systems, such as those presented later in this paper. However, our conditions are also more restrictive in the sense that we now require very specific forms for the expressions $\psi_1$ and $\psi_2$. But this is expected, since Theorem \ref{issthm} considers fixed-time ISS, which is a much stronger notion of stability. \hfill $\QEDB$
\end{remark}
}

\vspace{0.1cm}
We have shown that under suitable assumptions, the inequalities in \eqref{intcond} imply that system \eqref{xysys} is FxT ISS provided there is a sufficiently large time scale separation. {The proof of Theorem \ref{issthm} provides an efficient methodology for computing a somewhat conservative estimate of the required timescale separation, \tcb{i.e.,} $\varepsilon^*$, for FxT ISS. It can also be seen from the proof of Theorem \ref{issthm} that if the fast dynamics have no input, then the derived gain $\varrho$ in the FxT ISS bound \eqref{fxtISSbounddef} can be made independent of $\varepsilon$. This is further detailed in the following Corollary.
\begin{cor}\label{cor_unif}
    Consider \eqref{xysys} and suppose the conditions of Theorem \ref{issthm} are satisfied, but with $\rho_B\equiv 0$ and the vector field $g$ is independent of $u$. Then, there exists $\beta\in\mathcal{GKL}$, $\varrho\in\mathcal{K}$ and $\varepsilon^*>0$ such that the following holds 
    % \begin{equation}\label{cgklbd}
    %     \left\lvert\begin{bmatrix}
    %         x(t)\\ y(t)
    %     \end{bmatrix}\right\rvert\le\beta\left(\left\lvert\begin{bmatrix}
    %         x_0\\ y_0
    %     \end{bmatrix}\right\rvert,t\right)+\varrho(|u|_\infty),
    % \end{equation}
    \begin{equation}\label{cgklbd}
        \left\lvert s(t)\right\rvert\le\beta\left(\left\lvert s_0\right\rvert,t\right)+\varrho(|u|_\infty),
    \end{equation}
    for all $t\ge 0$, $u\in\mathcal{L}_\infty^p$, and $\varepsilon\in (0, \varepsilon^*)$, where $s(t):=[x(t), y(t)]^\top$.
\end{cor}
\begin{proof}
    % With the composite Lyapunov function candidate \eqref{compl}, we can follow the same steps in the proof of Theorem \ref{issthm} and arrive at $\dot{\Psi}_\zeta
    %     \le-\nu(\zeta)\tilde{V}^2{(x)}-\omega_\varepsilon(\zeta)\tilde{W}^2{(x,y)}+\tilde{\rho}_{\zeta}(|u|),$
    % with $\tilde{\rho}_{\zeta}(|u|)=\zeta\left(\rho_R(|u|)+\rho_1(|u|)\right)+(1-\zeta)\rho_2(|u|).$
    % We can find $\zeta^*\in(0,1)$ such that $\nu^*:=\nu(\zeta^*)>0$, and then we can find $\varepsilon^*>0$ such that $\omega_\varepsilon(\zeta^*)>\nu^*$ whenever $\varepsilon\in(0,\varepsilon^*)$. For these $\varepsilon\in(0, \varepsilon^*)$ we obtain $\dot{\Psi}_{\zeta^*}\le -\dfrac{\nu^*}{2}\left(\Psi_{\zeta^*}^{\gamma_1}+2^{1-\gamma_2}\Psi_{\zeta^*}^{\gamma_2}\right)+\tilde{\rho}_{\zeta^*}(|u|),$ where $\gamma_1\in \left[\max\left\{a_1, b_1\right\},1\right)$ and $\gamma_2\in\left(1,\min\left\{a_2, b_2\right\} \right]$. Moreover, we observe that the upper bound on $\dot{\Psi}_{\zeta^*}$ is independent of $\varepsilon$, which establishes the result.
    We follow the same steps from the proof of Theorem \ref{issthm}, but with $\tilde{\rho}_{\zeta}(s)=\zeta\left(\rho_R(s)+\rho_1(s)\right)+(1-\zeta)\rho_2(s)$ instead of $\rho_{\varepsilon, \zeta}$. This results in an upper bound on $\dot{\Psi}_{\zeta^*}$ that holds uniformly for $\varepsilon\in (0, \varepsilon^*)$.
\end{proof}
}
\vspace{0.1cm}

{While the results of Theorem \ref{issthm} and Corollary \ref{cor_unif} hold for system \eqref{xysys},} the question of whether \eqref{sys} is FxT ISS may also be of interest. In other words: we ask if the transformation $y=z-h(x)$ preserves the FxT ISS property. Fortunately, as long as the \tcb{quasi-steady} state map satisfies a mild boundedness assumption, \eqref{xysys} being FxT ISS implies \eqref{sys} is FxT ISS. This is further detailed in the following result.
\begin{thm}\label{zyeq}
    Suppose \eqref{xysys} is FxT ISS and $|h(x)|\le \tilde{\alpha}(|x|)$ for some $\tilde{\alpha}\in\mathcal{K}$, then \eqref{sys} is also fixed time ISS. 
    % Furthermore, if $\Psi(x,y)$ is a FxT ISS Lyapunov function for \eqref{xysys}, then $\tilde{\Psi}(x,z):=\Psi(x, z-h(x))$ is a FxT ISS Lyapunov function for \eqref{sys}. 
    \QEDB 
\end{thm}
\begin{proof}
    Since \eqref{xysys} is FxT ISS, there exists $\beta\in\mathcal{GKL}$ and $\varrho\in\mathcal{K}$ such that
    \begin{equation}\label{gklbd}
        \left\lvert\begin{bmatrix}
            x(t)\\ z(t)-h(x(t))
        \end{bmatrix}\right\rvert\le\beta\left(\left\lvert\begin{bmatrix}
            x_0\\ y_0
        \end{bmatrix}\right\rvert,t\right)+\varrho(|u|_\infty),
    \end{equation}
    where $\beta(r,t)=0$ for each $t>T(r)$, and $T$ is a continuous function that satisfies $\sup_{r\ge 0} T(r)<\infty$ and $T(0)=0$. Without loss of generality, we can assume $\beta(\cdot, t)$ is non-decreasing for each $t\ge 0$. Note that \eqref{gklbd} implies
    \begin{subequations}\label{xzkl}
        \begin{align}
            |x(t)|&\le \beta\left(\left\lvert\begin{bmatrix}
            x_0\\ y_0
        \end{bmatrix}\right\rvert,t\right)+\varrho(|u|_\infty)\\
        |z(t)|&\le \beta\left(\left\lvert\begin{bmatrix}
            x_0\\ y_0
        \end{bmatrix}\right\rvert,t\right)+\varrho(|u|_\infty)+|h(x(t))|.
        \end{align}
    \end{subequations}
    But we also have
    \begin{align*}
        \beta\left(\left\lvert\begin{bmatrix}
            x_0\\ y_0
        \end{bmatrix}\right\rvert,t\right)=\beta\left(\left\lvert\begin{bmatrix}
            x_0\\ z_0-h(x_0)
        \end{bmatrix}\right\rvert,t\right)\le\overline{\beta}\left(\left\lvert\begin{bmatrix}
            x_0\\ z_0
        \end{bmatrix}\right\rvert,t\right),
    \end{align*}
    where 
    $\overline{\beta}(r,t):=\beta(2r,t)+\beta(2\tilde{\alpha}(r),t)\in\mathcal{GKL}$. Moreover,
    \begin{align}
        |h(x(t))|\le \tilde{\alpha}(|x(t)|)\le\tilde{\beta}\left(\left\lvert\begin{bmatrix}
            x_0\\ z_0
        \end{bmatrix}\right\rvert,t\right)+\tilde{\varrho}(|u|_\infty),\label{hbd}
    \end{align}
    where $\tilde{\beta}(r,t):=\tilde{\alpha}\left(2\overline{\beta}\left(r,t\right)\right)\in\mathcal{GKL}$, and 
            $\tilde{\varrho}(\cdot):=\tilde{\alpha}(2{\varrho}(\cdot))\in\mathcal{K}$. Combining \eqref{xzkl} and \eqref{hbd} yields the result.
\end{proof}
\subsection{A Stylized Example}
To illustrate our results, we first consider a singularly perturbed system with scalar reduced and boundary layer systems. In particular, consider the plant
\begin{subequations}\label{ex}
    \begin{align}
        \dot{x}&=-\lceil z\rfloor^{r_1}-\lceil z\rfloor^{r_2}+u_1\\
        \varepsilon \dot{z}&=-\lceil z-x-u_1\rfloor^{q_1}-\lceil z-x-u_2\rfloor^{q_2}+u_1 u_2,
    \end{align}
\end{subequations}
where $\lceil \cdot\rfloor^{q}:=|\cdot|^q\text{sgn}(\cdot)$,  $x, z, u_1, u_2\in\mathbb{R}$, $0<q_1\le r_1<1$ and $1<r_2\le q_2$. System \eqref{ex} has the \tcb{quasi-steady} state $h(x)=x$, with reduced system:
\begin{equation}\label{exred}
    \dot{x}=-\lceil x\rfloor^{r_1}-\lceil x\rfloor^{r_2}+u_1,
\end{equation}
and boundary layer system
\begin{equation}\label{exbl}
    \dfrac{dy}{d\tau}=-\lceil y-u_1\rfloor^{q_1}-\lceil y-u_2\rfloor^{q_2}+u_1 u_2.
\end{equation}
Furthermore, it is easy to see that $h$ satisfies the class $\mathcal{K}$ bound assumption from Theorem \ref{zyeq}. To verify that \eqref{exred} and \eqref{exbl} satisfy Assumptions \ref{assump_rsys} and \ref{assump_blsys}, we will use the Lyapunov functions $V(x)=x^2$ and $W(y)=y^2$. For the reduced system \eqref{exred} we have:
\begin{align}
    \dfrac{\partial V}{\partial x}f(x, h(x), u)
    &\le -V^{{\tilde{r}_1}}{(x)}-V^{{\tilde{r}_2}}{(x)}+|u|^2,
    \label{toyr}
\end{align}
{where $\tilde{r}_i:=\frac12(r_i+1)$.}
Thus, Assumption \ref{assump_rsys} is satisfied.
% To proceed with the boundary layer dynamics, we first establish the following useful lemma:
% \begin{lemma}\label{exlem}
%     Given $\alpha>0$ and $x, u\in\mathbb{R}$ the following holds for $|x|> 2|u|$:
%     \begin{equation*}
%         x\lceil x+u \rfloor^\alpha> 2^{-\alpha}|x|^{\alpha+1}.
%     \end{equation*}
% \end{lemma}
% \vspace{0.1cm}
% \begin{proof}
%     First assume $x<0$, which implies $|u|<-\frac12 x$. Since $\lceil x \rfloor^\alpha$ is strictly increasing in $x$, we have
%     \begin{align*}
%         \lceil x+u \rfloor^\alpha< \left\lceil x-\frac12 x \right\rfloor^\alpha=2^{-\alpha}\lceil x \rfloor^\alpha.
%     \end{align*}
%     Multiplying by $x$ yields the result. Now, suppose $x>0$, which implies $-|u|>-\frac12 x$. Then we have
%     \begin{align*}
%         \lceil x+u\rfloor^\alpha\ge\lceil x-|u|\rfloor^\alpha>\left\lceil x-\frac12 x\right\rfloor^\alpha=2^{-\alpha}\lceil x\rfloor^\alpha.
%     \end{align*}
%     \vspace{0.1cm}
%     We again multiply by $x$ to obtain the result.
% \end{proof}

% \vspace{0.1cm}
Similarly, we can differentiate $W$ along the trajectories of \eqref{exbl} to obtain the following:
\begin{align*}
    \dot{W}
     &\le -2^{1-q_1}|y|^{q_1+1}-2^{1-q_2}|y|^{q_2+1}+2yu_1 u_2\\
     &\le -(2^{1-q_1}-2^{-q_2})|y|^{q_1+1}-2^{-q_2}|y|^{q_2+1}+2^{q_2}|u|^4\\
    &\le -(2^{-q_1}-2^{-1-q_2})W^{{\tilde{q}_1}}{(y)}-2^{-1-q_2}W^{{\tilde{q}_2}}{(y)},\\&~~~\quad \forall~|y|>\max\{{\varrho}_B^{-1}(2^{q_2}|u|^4), 2|u|\},
\end{align*}
where we used Lemma \ref{lemma_ex0} to obtain the first inequality, as well as $2yu_1u_2\leq c|y|^2+\frac{1}{c}|u_1u_2|^2$ (with $c=2^{-q_2}$) and $|u_i|\leq|u|$ to obtain the second inequality, and
\begin{equation}
-2^{-q_2}|y|^{q_1+1}-2^{-q_2}|y|^{q_2+1}+2^{-q_2}|y|^2\leq0
\end{equation}
to obtain the third inequality with $\tilde{q}_i=\frac12(q_i+1)$ and
\begin{equation*}
    \varrho_B(s):=(2^{-q_1}-2^{-1-q_2})|s|^{q_1+1}+2^{-1-q_2}|s|^{q_2+1}.
\end{equation*}
This implies that system \eqref{exbl} is FxT ISS uniformly in $x$, and hence the assumptions are satisfied. To verify if system \eqref{ex} is FxT ISS, it remains to check if the interconnection terms associated with system \eqref{ex} satisfy the interconnection conditions \eqref{intcond}. Let $\tilde{V}{(x)}=|x|^{{\tilde{r}_1}}+|x|^{{\tilde{r}_2}}$ and $\tilde{W}{(y)}=|y|^{{\tilde{q}_1}}+|y|^{{\tilde{q}_2}}$. Then, we can compute $I_1$:
\begin{align*}
    I_1&=2x\Bigl(-\lceil y+x\rfloor^{r_1}-\lceil y+x\rfloor^{r_2}+\lceil x\rfloor^{r_1}+\lceil x\rfloor^{r_2}\Bigl).
    % &\le 2|x|\left(2^{1-r_1} |y|^{r_1}+K_{1-r_2}\left(|y|^{r_2}+|y||x|^{r_2-1}\right)\right)\\
    % &\le \dfrac{1}{\underline{\alpha}_1(c)}\left(|x|^{r_1+1}+|x|^{r_2+1}\right)+\overline{\alpha}_1(c)\left(|y|^{r_1+1}+|y|^{r_2+1}\right)\\   
    % &\le c\tilde{V}^2{(x)}+2c^{-\sigma_1}\tilde{W}^2{(y)},
\end{align*}
By Lemma \ref{bounds}, we have
\begin{equation}\label{ex0_i1bd1}
    \left\lvert \sg{x}^{r_1}-\sg{y+x}^{r_1}\right\rvert\le 2^{1-r_1}|y|^{r_1},
\end{equation}
and
\begin{equation}\label{ex0_i1bd2}
    \left\lvert \sg{x}^{r_2}-\sg{y+x}^{r_2}\right\rvert\le K|y|\left(|x|^{r_2-1}+|y|^{r_2-1}\right),
\end{equation}
where $K=1+\max\{1, (r_2-1)2^{r_2-2}\}$ and we use the fact that $\sg{x}^r=\frac{x}{|x|^{1-r}}$ for each $r>0$. With \eqref{ex0_i1bd1} and \eqref{ex0_i1bd2}, we can use the triangle inequality to bound $I_1$ as follows
\begin{align}
    I_1&\le 2|x|\left(2^{1-r_1}|y|^{r_1}+K|y|\left(|x|^{r_2-1}+|y|^{r_2-1}\right)\right)\notag\\
    &=2^{2-r_1}|x||y|^{r_1}+2K|y||x|^{r_2}+2K|x||y|^{r_2}.\label{ex0_i1bdxy}
\end{align}
By Lemma \ref{amgm}, we also have that the following inequalities hold
\begin{subequations}\label{ex0amgm}
\begin{align}
    2^{2-r_1}|x||y|^{r_1}&\le 2^{2-r_1}\left(\hat{c}|x|^{r_1+1}+\hat{c}^{-\frac{1}{r_1}}|y|^{r_1+1}\right)\\
    2K|y||x|^{r_2}&\le 2K\left(\hat{c}|x|^{r_2+1}+\hat{c}^{-r_2}|y|^{r_2+1}\right)\\
    2K|x||y|^{r_2}&\le 2K\left(\hat{c}|x|^{r_2+1}+\hat{c}^{-\frac{1}{r_2}}|y|^{r_2+1}\right)
\end{align}
\end{subequations}
for all $\hat{c}>0$. Since $2\tilde{q}_1\le 2\tilde{r}_1\leq r_i+1\le 2\tilde{r}_2\le 2\tilde{q}_2$ for $i=1,2$, we can use Lemma \ref{lem_sandw} in the Appendix (and the fact that $(a+b)^2\geq a^2+b^2,~\forall~a,b\geq0$) to obtain the following
\begin{subequations}\label{ex0_i1lft}
\begin{align}
    |x|^{r_i+1}\le |x|^{2\tilde{r}_1}+|x|^{2\tilde{r}_2}&\le \tilde{V}^2(x)\\
    |y|^{r_i+1}\le |y|^{2\tilde{q}_1}+|y|^{2\tilde{q}_2}&\le \tilde{W}^2(y)
\end{align}
\end{subequations}
for $i=1,2$. For fixed $\hat{c}>0$, let $c=(2^{2-r_1}+4K)\hat{c}$, and we define the function $\overline{\alpha}:(0,\infty)\to (0, \infty)$ by
\begin{align*}
    \overline{\alpha}(s)=2^{2-r_1}s^{-\frac{1}{r_1}}+2K(s^{-r_2}+s^{-\frac{1}{r_2}}).
\end{align*}
We let $\tcb{\sigma_1:=\max\{\frac{1}{r_1}, r_2\}}$ and notice that for all $s\in (0,1)$, the function $\overline{\alpha}$ satisfies the following bound
\begin{equation}\label{boundsigmastyel}
    \overline{\alpha}(s)\le (2^{2-r_1}+4K)s^{-\sigma_1}.
\end{equation}
By combining \eqref{ex0_i1bdxy}, \eqref{ex0amgm}, \eqref{ex0_i1lft}, and \eqref{boundsigmastyel} we obtain
\begin{align*}
    I_1&\le c \tilde{V}^2(x)+\overline{\alpha}(\hat{c})\tilde{W}^2(y)\\
    &=c \tilde{V}^2(x)+\overline{\alpha}\left(\frac{c}{2^{2-r_1}+4K}\right)\tilde{W}^2(y)\\
    &\le c \tilde{V}^2(x)+\tcb{(2^{2-r_1}+4K)^{1+\sigma_1}}c^{-\sigma_1}\tilde{W}^2(y)
\end{align*}
for all $c\in (0,\tcb{2^{2-r_1}+4K})$. By taking $c>0$ sufficiently small we can satisfy conditions \eqref{intu1} and \eqref{vcond}. Thus, it remains to check if $I_2$ also satisfies \eqref{intu2}. Indeed, by computing $I_2$ we obtain:
\begin{align*}
    I_2&=2y\left(\lceil y+x\rfloor ^{r_1}+\lceil y+x\rfloor ^{r_2}-u_1\right)\\
    &\le 2|y|\left(|y+x|^{r_1}+|y+x|^{r_2}+|u_1|\right)\\
    &\le 2|y|\left(|y|^{r_1}+|x|^{r_1}+2^{r_2-1}(|y|^{r_2}+|x|^{r_2})+|u_1|\right)\\
    &=2(|y|^{r_1+1}+|y||x|^{r_1}+|y||u_1|)+2^{r_2}(|y|^{r_2+1}+|y||x|^{r_2})
    % &\le 
    % 2|y|\left(|y|^{r_1}+|x|^{r_1}+2^{r_2-1}\left(|y|^{r_2}+|x|^{r_2}\right)+|u|\right)\\
    % &\le 
    % \dfrac{1}{\underline{\alpha}_2(c)}\left(|x|^{r_1+1}+|x|^{r_2+1}\right)+\overline{\alpha}_2(c)\left(|y|^{r_1+1}+|y|^{r_2+1}\right)\\&~~~+2|y|^{r_1+1}+2^{r_2}|y|^{r_2+1}+2|y||u|\\
    % \\&\le c\tilde{V}^2{(x)}+(2c^{-\sigma_2}+3+2^{r_2})\tilde{W}^2{(y)}+|u|^2,
\end{align*}
By Lemma \ref{amgm}, we have
\begin{align*}
    |y||x|^{r_1}&\le |x|^{r_1+1}+|y|^{r_1+1}\\
    |y||x|^{r_2}&\le |x|^{r_2+1}+|y|^{r_2+1}.
\end{align*}
Then, we can leverage \eqref{ex0_i1lft} to obtain
\begin{align*}
    I_2&\le 4|y|^{r_1+1}+2|x|^{r_1+1}+|y|^2+|u|^2+2^{r_2+1}|y|^{r_2+1}\\&~~~+2^{r_2}|x|^{r_2+1}\\
    &\le \tcb{(2+2^{r_2})}\tilde{V}^2(x)+\tcb{(5+2^{r_2+1})}\tilde{W}^2(y)+|u|^2.
\end{align*}
We conclude that the conditions of Theorem \ref{issthm} and \ref{zyeq} are satisfied, and hence there exists $\varepsilon^*$ such that the singularly perturbed system \eqref{ex} is FxT ISS for $\varepsilon\in (0, \varepsilon^*)$. We simulate the system using $r_1=\frac25, r_2=\frac65, q_1=\frac13, q_2=\frac97$ and the disturbances $u_1(t)=e^{\sin{t}}, u_2(t)= \sin(19\log(t+1))-0.21$, {where we clearly have $u(t)=[u_1(t), u_2(t)]^\top \in \mathcal{L}_\infty^2$. The trajectories of this system with and without the input are shown in Figure \ref{fxt1plot}, illustrating the FxT ISS property}
\begin{figure}[t!]
  \centering \includegraphics[width=0.45\textwidth]{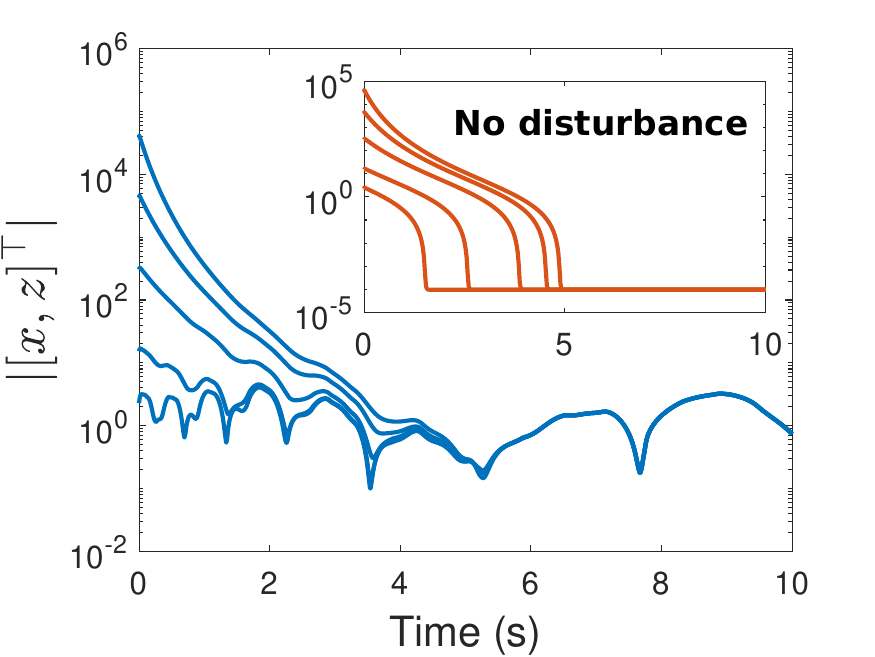}
    \caption{Trajectories of system \eqref{ex} {with and without the disturbance $u(t)$, where we use $\varepsilon=0.01$. The theoretically computed settling time bound of $18.15$ obtained using \eqref{settime} is conservative compared to the observed trajectories.}} \label{fxt1plot}
    \vspace{-0.4cm}
\end{figure}
\section{Fixed-Time Feedback Optimization with Time-varying Costs}\label{sec_ex}
In this section, we leverage the results of Theorems \ref{issthm}-\ref{zyeq} to study a practical problem of interest in the context of singular perturbations: feedback optimization under slowly-varying cost functions \cite{colombino2019online}. In contrast to the existing asymptotic results \cite{hauswirth2020timescale,bianchin2022online,9540998, 10189107}, we introduce an optimization-based controller able to achieve \tcb{closed-loop} FxT stability via non-Lipschitz feedback. In particular, we consider plants of the form
\begin{equation}\label{explant}
    \dot{z}=g(\hat{x},z),
\end{equation}
where $z\in\mathbb{R}^m$ is the state, $\hat{x}\in\mathcal{L}_\infty^n$ is a measurable and bounded control input, and $g:\re^n\times\re^m\to\mathbb{R}^m$ is a continuous function satisfying the following condition:
\begin{assumption}\label{assump_g}
    There exists a continuously differentiable, globally Lipschitz mapping $h:\mathbb{R}^n\to\mathbb{R}^m$ such that $g(\hat{x}, h(\hat{x}))=0$ for all $\hat{x}\in\mathbb{R}^n$. Moreover, there exists a $\mathcal{C}^1$ function $W:\mathbb{R}^m\to\mathbb{R}_{\ge 0}$ and $c_1, c_2>0$ such that
    \begin{subequations}\label{sandwichw}
    \begin{align}
        c_1|y|^2&\le W(y)\le c_2 |y|^2,\\
        \nabla W(y)^\top g(\hat{x}, y+h(\hat{x}))&\le -\kappa_1 W^{b_1}{(y)}-\kappa_2 W^{b_2}{(y)},
    \end{align}
    \end{subequations}
    for all $y\in\mathbb{R}^m$, where $\kappa_1, \kappa_2>0$, $b_1\in (0,1)$ and $b_2>1$. There also exists $\eta>0$ such that the function $W$ satisfies
        \begin{align}\label{assump_dw}
            \left\lvert \nabla W(y)\right\rvert&\le \eta |y|,
        \end{align}
    for all $y\in\mathbb{R}^{m}$.\QEDB
\end{assumption}

\vspace{0.1cm}
{Assumption \ref{assump_g} is the ``fixed-time'' version of standard open-loop stability assumptions considered in the setting of feedback optimization \cite{khalil, hauswirth2020timescale}. In particular, by taking $y=z-h(\hat{x})$ to quantify the deviation of $z$ from its steady-state approximation, the conditions of Assumption \ref{assump_g} simply ask that such deviation converges to zero by a fixed-time, for each fixed $\hat{x}\in\mathbb{R}^n$. For a general class of linear and nonlinear plants, this property can be achieved using different types of non-smooth controllers that combine super-linear and sub-linear feedback \cite{andrieu2008homogeneous}, homogeneity tools \cite{10704051}, or the implicit Lyapunov technique \cite{POLYAKOV2015332}, to name just a few examples.
}

\vspace{0.15cm}
Our primary goal is to design a control law on the input $\hat{x}$ that stabilizes \eqref{explant}\tcb{, in a fixed time, at} the solution of the following time-varying optimization problem
\begin{subequations}
\begin{align}\label{timeopt}
    &\min_{\hat{x},z}\ \phi_\theta(\hat{x},z)\\
    &\text{subject~ to:}~~z=h(\hat{x}).\label{optcon}
\end{align}
\end{subequations}

\begin{figure}[t!]
  \centering \includegraphics[width=0.35\textwidth]{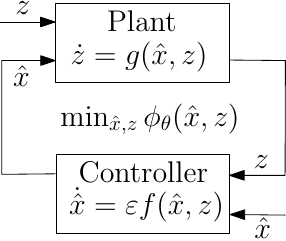}
    \caption{A block diagram illustrating the interconnection \eqref{exint}.} \label{fxtblock}
    \vspace{-0.4cm}
\end{figure}
{where the time variation on the cost functions $\phi_\theta:\mathbb{R}^n\times\mathbb{R}^m\to\mathbb{R}$ is induced by a dynamic parameter $\theta\in\mathbb{R}^q$ that evolves according to the dynamics
\begin{equation}   \dot{\theta}=\varepsilon\varepsilon_0\Pi(\theta),\quad\theta\in\Theta,\label{dtheta}
\end{equation}
where $\varepsilon_0\geq0$, and $\varepsilon>0$ is a small parameter that captures the rate of change of $\theta$, and $\Pi:\mathbb{R}^q\to\mathbb{R}^q$ and $\Theta\subset\mathbb{R}^q$ satisfy the following mild conditions, which can be used to cover a broad class of time-varying signals $t\mapsto\theta(t):$
\begin{assumption}\label{assumpt}
    The function $\Pi(\cdot)$ is Lipschitz continuous, and the set $\Theta$ is compact and forward invariant under the dynamics \eqref{dtheta}.\QEDB
\end{assumption}}
\vspace{0.1cm}

We can observe that, by substituting \eqref{optcon} into \eqref{timeopt}, we arrive at the unconstrained parameterized optimization problem:
\begin{equation}\label{opt}
    \min_{\hat{x}}\Phi_\theta(\hat{x}),
\end{equation}
where $\Phi_\theta(\hat{x}):=\phi_\theta(\hat{x}, h(\hat{x}))$. To guarantee that \eqref{opt} is a well-defined optimization problem with a unique solution for each $\theta\in\mathbb{R}^q$, we consider cost functions that satisfy the following assumptions,  {which are fairly standard in the time-varying feedback optimization literature \cite{9540998, 10189107}}:
\begin{assumption}\label{assump_phi}
    {The function $\hat{x}\mapsto\Phi_\theta(\hat{x})$ is $L$-smooth and $\kappa$-strongly convex, uniformly in $\theta$. Moreover, there exists a $\mathcal{C}^1$ function $\varphi:\mathbb{R}^q\to\mathbb{R}^n$ such that $\varphi(\theta)=\argmin_{\hx} \Phi_\theta(\hx)$.}\QEDB
\end{assumption}
% \vspace{0.1cm}
% Under this assumption, we obtain the following result:
% \begin{lemma}\label{lemp}
%     If the functions $\Phi_\theta(\cdot)$ satisfy Assumption \ref{assump_phi}, the following hold
%     \begin{subequations}\label{lemeq}
%     \begin{align}
%         \hx^\top \nabla\Phi_\theta(\hx+\varphi(\theta))&\ge \frac{1}{L}|\nabla\Phi_\theta(\hx+\varphi(\theta))|^2\\
%         |\nabla\Phi_\theta(\hx+\varphi(\theta))|&\ge \kappa|\hx|,
%     \end{align}
% \end{subequations}
% for all $\hx\in\mathbb{R}^n$ and $\theta\in\mathbb{R}^q$.\QEDB
% \end{lemma}
% \vspace{0.1cm}

\vspace{0.1cm}

Note that when $\varepsilon_0 = 0$ in \eqref{assumpt}, the parameter $\theta$ remains constant for all time, yielding a constant solution $\varphi(\theta)$ to \eqref{opt}. In contrast, when $\varepsilon_0 \gg 0$, the function $t \mapsto \varphi(\theta(t))$ may exhibit fast time variations that are difficult to track without additional information about the functions $\Pi$, $h$, and $\phi_{\theta}$. Therefore, to address the optimization problem \eqref{timeopt} using real-time gradient feedback, we can regard $\varepsilon_0 \Pi(\theta)$ as the ``input'' to the system and study FxT ISS with respect to the \emph{tracking error} of $\varphi(\theta(\cdot))$.

% Ideally, we would like to solve \eqref{opt} while stabilizing \eqref{explant} in fixed-time. However, since \eqref{opt} is a time-varying problem, the desired alternative would be to achieve \emph{fixed-time tracking} of the time-dependent solution of \eqref{opt}
\subsection{Fixed-Time Gradient-Based Feedback}
Given $\xi_1\in (0,1)$ and $\xi_2<0$, we define the following function $\mathcal{F}_{\xi_1,\xi_2}:\mathbb{R}^p\to\mathbb{R}^p$:
\begin{equation}
    \mathcal{F}_{\xi_1, \xi_2}(x)=\frac{x}{|x|^{\xi_1}}+\frac{x}{|x|^{\xi_2}},
\end{equation}
{which is continuous at $x=0$, see \cite{fixed_time}. To solve \eqref{opt} in fixed-time, we draw inspiration from the asymptotic counterpart \cite{hauswirth2020timescale} and propose a \emph{fixed-time} gradient flow on $\Phi_\theta(\hat{x})$ with time scale separation:
\begin{equation}\label{rgflow}
    \dot{\hat{x}}=-\varepsilon\fx(\hat{P}_\theta(\hx, h(\hx))),
\end{equation}
where
\begin{equation}\label{phat}
    \hat{P}_\theta(\hx, z):=H(\hx)^\top\nabla\phi_\theta(\hx,z),\quad H(x)^\top=[\mathbb{I}_n\quad \bj_{h}(x)^\top].
\end{equation}
It can be verified, using the chain rule, that 
\begin{equation*}
    \hat{P}_\theta(\hat{x}, h(\hat{x})) = \nabla \Phi_\theta(\hat{x}),
\end{equation*}
and hence, for each $\theta$, when the plant dynamics \eqref{explant} are negligible, the dynamics \eqref{rgflow} converge to the solution of \eqref{opt} in fixed time~\cite{fixed_time,nonsmoothesc}. Since in our setting the dynamics \eqref{explant} cannot be neglected, we can obtain a real-time feedback controller by replacing $h(\hx)$ in \eqref{rgflow} with the measured value of $z$ to obtain the following closed-loop system:
\begin{subequations}\label{exint}
    \begin{align}
        \dot{z}&=g(\hx,z)\\
        \dot{\hat{x}}&=-\varepsilon\fx(\hat{P}_\theta(\hx, z)).\label{exintb}
    \end{align}
\end{subequations}}
\begin{figure*}[t!]
  \centering
    \includegraphics[width=0.35\textwidth]{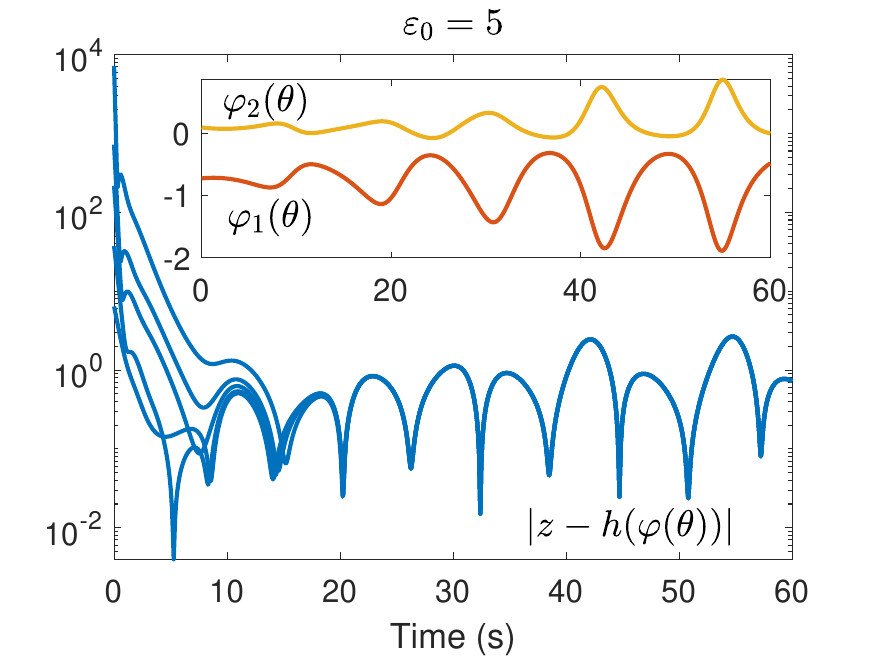}\hspace{-0.5cm}\includegraphics[width=0.35\textwidth]{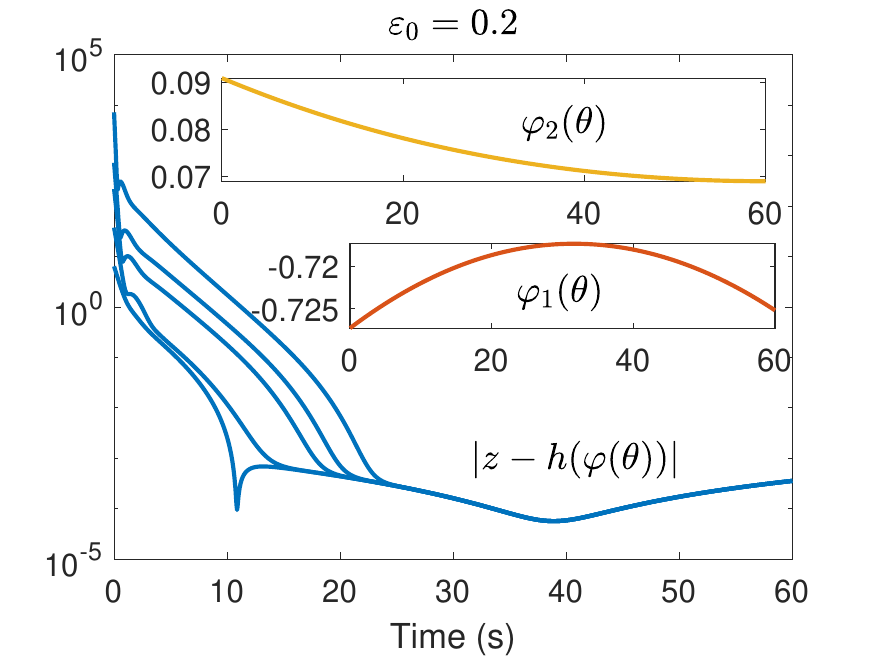}\hspace{-0.5cm}\includegraphics[width=0.35\textwidth]{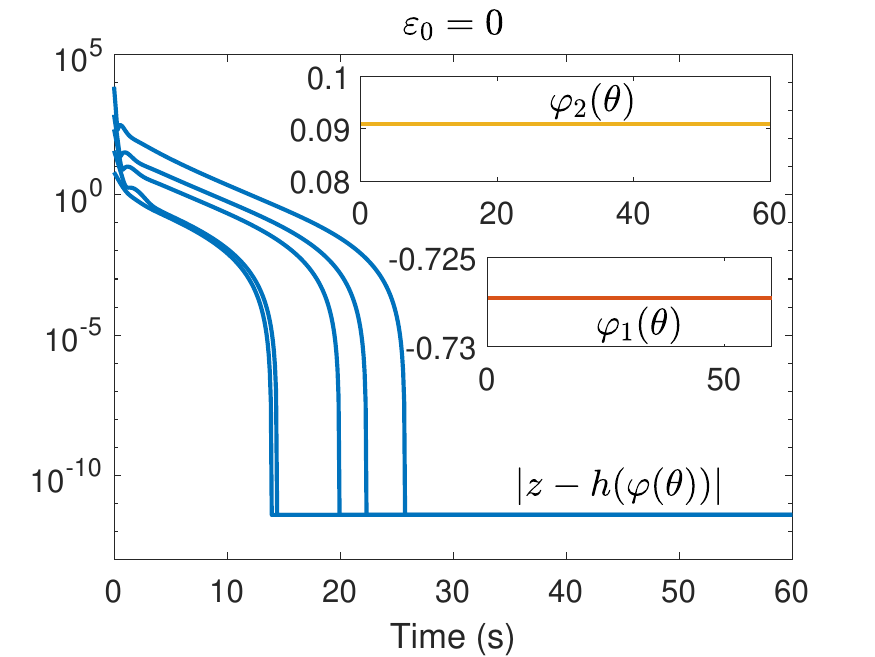}
    \caption{Trajectories of the FxT time-varying feedback optimization example, with $\varepsilon_0=5,0.2, 0$ and varying initial conditions.}\label{ex2plot}  
    \vspace{-0.2cm}
\end{figure*}
We study system \eqref{exint} under the following mild Lipschitz assumption, which is standard in the literature \cite{hauswirth2020timescale}:
\begin{assumption}\label{asp}
    For the function $\hat{P}_\theta(x,z)$ defined in \eqref{phat}, there exists $\ell>0$ such that 
    \begin{equation}\label{lipp}
        |\hat{P}_\theta(\hx,z')-\hat{P}_\theta(\hx,z)|\le \ell |z'-z|,
    \end{equation}
    for all $\hx\in\mathbb{R}^n$, $z',z\in\mathbb{R}^m$ and $\theta\in\mathbb{R}^q$.\QEDB
\end{assumption}

\vspace{0.1cm}
To put system \eqref{exint} into the form \eqref{sys}, let $x=\hx-\varphi(\theta)$ and $\tau=\varepsilon t$, which leads to the following dynamics in the $\tau$-time scale:
\begin{subequations}\label{exinth}
    \begin{align}
        \varepsilon\frac{dz}{d\tau}&=g(x+\vt,z)\label{zdynamics}\\
        \frac{dx}{d\tau}&=-\fx(\hat{P}_\theta(x+\vt, z))-\bj_\varphi(\theta)u(\tau)\label{xdynamics}\\
        \frac{d\theta}{d\tau}&=\varepsilon_0\Pi(\theta),\quad\theta\in\Theta\label{thetadynamics},
    \end{align}
\end{subequations}
where $u(\tau)=\varepsilon_0\Pi(\theta(\tau))$ can be thought of as the ``input" in \eqref{xdynamics}. Note that, since $\Theta$ is compact and forward invariant, and the trajectories of \eqref{thetadynamics} are restricted to evolve in $\Theta$, we only need to consider the stability properties of system \eqref{zdynamics}-\eqref{xdynamics}. The following result establishes FxT ISS for this system: 
%
%We will show that this system is FxT ISS with $\dot{\theta}$ being treated as the ``disturbance", so slower variations of $\theta$ will result in a more accurate $\mathcal{GKL}$ bound on the flows of \eqref{exint}. Since we assume $\Theta$ is compact, meaning $\Pi(\theta)$ is bounded on $\Theta$, we can also view $\varepsilon_0$ as the ``input". 
%
%{Since we view $\theta$ as being one of the system states, we aim to establish FxT ISS of the set $\{[h(\varphi(\theta)),0,\theta]^\top : \ \theta\in\Theta \}$ for the $[z,x,\theta]^\top$ system. This is a compact set, and its stability properties can be addressed by generalizing the notion of FxT ISS to include sets (i.e, replacing $|x|$ in Definitions \ref{fxtissdef} and \ref{fxtissv} with $|x|_{\mathcal{A}}$, where $|x|_{\mathcal{A}}=\inf_{y\in\mathcal{A}}|x-y|$).} 
%
%For this interconnection we can state the following result:

\vspace{0.1cm}
 \begin{thm}\label{exproof}
  Suppose that Assumptions \ref{assump_g}-\ref{assump_phi} hold, and consider the singularly perturbed system \eqref{exinth}. Then, for all $\xi_1\in\left(0,\min(2-2b_1,1)\right)$ and $\xi_2\in \left(2-2b_2,0\right)$, there exists $\beta\in\mathcal{GKL}$, $\varrho\in\mathcal{K}$ and $\varepsilon^*>0$ such that
     \begin{align*}
         &\left\lvert\begin{bmatrix}
             x(\tau)\\ z(\tau)-h(x(\tau)+\varphi(\theta(\tau)))
         \end{bmatrix}\right\rvert\\&~~~~~~~~~\le \beta\left(\left\lvert\begin{bmatrix}
             x_0\\ z_0-h(x_0+\varphi(\theta_0))
    \end{bmatrix}\right\rvert,\tau\right)+\varrho(|u(\tau)|_\infty),
     \end{align*}
     for all $\varepsilon\in(0, \varepsilon^*)$, $\tau\ge 0$ and $\varepsilon_0\ge 0$. \QEDB 
 \end{thm}
 \begin{proof} Consider the $\tau$-time scale system \eqref{zdynamics}. This system has a quasi-steady state $h(x+\vt)$, which results in the following reduced system:
 \begin{align}\label{rsys_ex}
    \frac{dx}{d\tau}&=-\fx(\nabla\Phi_\theta(x+\varphi(\theta)))-\varepsilon_0\bj_\varphi(\theta)\Pi(\theta).
\end{align}
Let $M=\sup_{\theta\in\Theta}|\bj_\varphi(\theta)|$, which is bounded since $\Theta$ is compact. Moreover, from Assumption \ref{assump_phi}, we have $\hx^\top \nabla\Phi_\theta(\hx+\varphi(\theta))\ge \frac{1}{L}|\nabla\Phi_\theta(\hx+\varphi(\theta))|^2$ and $|\nabla\Phi_\theta(\hx+\varphi(\theta))|\ge \kappa|\hx|$. Then, with the Lyapunov function $V(x)=|x|^2$, we can use Lemma \ref{amgm} to obtain
\begin{align*}
    &\frac{dV}{d\tau}\le -\frac{2}{L}|\nabla\Phi_\theta(x+\varphi(\theta))|^{2-\xi_1}-\frac{2}{L}|\nabla\Phi_\theta(x+\varphi(\theta))|^{2-\xi_2}\\&~~~+\tilde{c}|x|^2+\frac{M^2}{\tilde{c}}|\varepsilon_0\Pi(\theta)|^2\\
    &\le -\frac{\kappa^{2-\xi_1}}{L}V^{1-\frac12\xi_1}{(x)}-\frac{\kappa^{2-\xi_2}}{L}V^{1-\frac12\xi_2}{(x)}+\frac{M^2}{\tilde{c}}|\varepsilon_0\Pi(\theta)|^2,
\end{align*}
where $0<\tilde{c}<\frac{1}{L}\min\left\{{\kappa^{2-\xi_1}}, {\kappa^{2-\xi_2}}\right\}$, which shows that Assumption \ref{assump_rsys} is satisfied. Next, let $y=z-h(x+\vt)$, which leads to the following dynamics:
\begin{align}\label{ysys_ex}
    \frac{dy}{d\tau}&=\frac{1}{\varepsilon}g(x+\vt, y+h(x+\vt))\\&~~~-\bj_{h}(x+\vt)\left(\frac{dx}{d\tau}+\varepsilon_0\bj_\varphi(\theta)\Pi(\theta)\right).\notag
\end{align}
This yields the following boundary layer system:
\begin{equation}\label{blsys_ex}
    \dot{y}=g(x+\vt, y+h(x+\vt)).
\end{equation}
It is easy to verify that Assumption \ref{assump_blsys} holds with Lyapunov function $W(y)$ obtained from Assumption \ref{assump_g}, so it remains to check the interconnection conditions \eqref{intu1}-\eqref{intu2}. Indeed, if we denote $P_\theta(x,y):=\hat{P}_\theta(x+\vt, y+h(x+\vt))$, we have:
\begin{align*}
    |I_1|&\le 2|x|\bigg(\left\lvert \frac{P_\theta(x,y)}{|P_\theta(x,y)|^{\xi_1}}-\frac{P_\theta(x,0)}{|P_\theta(x,0)|^{\xi_1}}\right\rvert\\&~~~+\left\lvert \frac{P_\theta(x,y)}{|P_\theta(x,y)|^{\xi_2}}-\frac{P_\theta(x,0)}{|P_\theta(x,0)|^{\xi_2}} \right\rvert\bigg).
\end{align*}
We denote $\Delta_\theta(x,y)=P_\theta(x,y)-P_\theta(x,0)$. Recall that from Assumption \ref{asp}, we have $|\Delta_\theta(x,y)|\le \ell |y|$. Hence, we can use Lemma \ref{bounds} to obtain
\begin{align}
    \left\lvert \frac{P_\theta(x,y)}{|P_\theta(x,y)|^{\xi_1}}-\frac{P_\theta(x,0)}{|P_\theta(x,0)|^{\xi_1}}\right\rvert&\le 2^{\xi_1}|\Delta_\theta(x,y)|^{1-\xi_1}\notag\\
    &\le 2^{\xi_1}\ell^{1-\xi_1}|y|^{1-\xi_1}\label{ex2i1bd1}.
\end{align}
We can again leverage Lemma \ref{bounds} to obtain
\begin{align}
    &\left\lvert \frac{P_\theta(x,y)}{|P_\theta(x,y)|^{\xi_2}}-\frac{P_\theta(x,0)}{|P_\theta(x,0)|^{\xi_2}} \right\rvert\notag \\&\le K |\Delta_\theta(x,y)|\left(|P_\theta(x,0)|^{-\xi_2}+|\Delta_\theta(x,y)|^{-\xi_2}\right)\notag\\
    &\le K\ell |y|\left(L^{-\xi_2}|x|^{-\xi_2}+\ell^{-\xi_2}|y|^{-\xi_2}\right)\label{ex2i1bd2},
\end{align}
where the bottom inequality uses the fact that $|P_\theta(x,0)|=|\nabla \Phi(x+\varphi_\theta(\theta))|\le L|x|$. We can combine \eqref{ex2i1bd1} and \eqref{ex2i1bd2} to obtain
\begin{align}
    |I_1|
    &\le 2^{\xi_1+1}\ell^{1-\xi_1}|x||y|^{1-\xi_1}\notag\\&~~~+\tcb{2K\ell}|x||y|\left(\ell^{-\xi_2}|y|^{-\xi_2}+L^{-\xi_2}|x|^{-\xi_2}\right)\label{ex1i1bd}
\end{align}

By Lemma \ref{amgm}, we have
\begin{subequations}\label{ex1amgm}
\begin{align}
    \delta_1|x||y|^{1-\xi_1}&\le {c}|x|^{2-\xi_1}+\delta_1^{1+\sigma_1}{c}^{-\sigma_1}|y|^{2-\xi_1}\\
    \delta_2|x||y|^{1-\xi_2}&\le {c}|x|^{2-\xi_2}+\delta_2^{1+\sigma_2}{c}^{-\sigma_2}|y|^{2-\xi_2}\\
    \delta_3|x|^{1-\xi_2}|y|&\le {c}|x|^{2-\xi_2}+\delta_3^{1+\frac{1}{\sigma_2}}{c}^{-\frac{1}{\sigma_2}}|y|^{2-\xi_2}
\end{align}
\end{subequations}
for all ${c}>0$, where $\sigma_i=\frac{1}{1-\xi_i}$, $\delta_1=2^{\xi_1+1}\ell^{1-\xi_1}, \delta_2=2K\ell^{1-\xi_2}$, and $\delta_3=2K\ell L^{-\xi_2}$. Let $\tilde{V}(x)=|x|^{1-\frac12 \xi_1}+|x|^{1-\frac12\xi_2}$ and $\tilde{W}(y)=W^{\frac12 b_1}(y)+W^{\frac12 b_2}(y)$. Since
\begin{equation}
    2b_1\le 2-\xi_1\le2-\xi_i\le 2-\xi_2\le 2b_2
\end{equation}
holds for $i=1,2$, we also have that
\begin{subequations}\label{ex1lfbd}
    \begin{align}
        |x|^{2-\xi_i}&\le |x|^{2-\xi_1}+|x|^{2-\xi_2}\le \tilde{V}^2(x)\\
        |y|^{2-\xi_i}&\le c_1^{\frac12\xi_i-1}W^{1-\frac12\xi_1}(y)\le c_1^{\frac12\xi_i-1}\tilde{W}^2(y) \label{ex1ylf}
    \end{align}
\end{subequations}
for $i=1,2$. 

We can then combine \eqref{ex1i1bd}, \eqref{ex1amgm}, \eqref{ex1lfbd} to obtain
\begin{align*}
|I_1|&\le c(|x|^{2-\xi_1}+2|x|^{2-\xi_2})+\delta_1^{1+\sigma_1}{c}^{-\sigma_1}|y|^{2-\xi_1}\\&~~~+\delta_2^{1+\sigma_2}{c}^{-\sigma_2}|y|^{2-\xi_2}+\delta_3^{1+\frac{1}{\sigma_2}}{c}^{-\frac{1}{\sigma_2}}|y|^{2-\xi_2}
\end{align*}
Next, using 
\begin{align*}
m&:=\max\{\delta_1^{1+\sigma_1}, \delta_2^{1+\sigma_2}, \delta_3^{1+\frac{1}{\sigma_2}}\}\\
&=\max\{(2^{\xi_1+1}\ell^{1-\xi_1})^{1+\sigma_1}, \tcb{(2K\ell^{1-\xi_2})^{1+\sigma_2}},\tcb{(2K\ell L^{-\xi_2})^{1+\frac{1}{\sigma_2}}}\}
\end{align*}
we obtain:
\begin{align*}
|I_1|&\le c(|x|^{2-\xi_1}+2|x|^{2-\xi_2})\\
&~~~~+m\Big({c}^{-\sigma_1}|y|^{2-\xi_1}+{c}^{-\sigma_2}|y|^{2-\xi_2}+{c}^{-\frac{1}{\sigma_2}}|y|^{2-\xi_2}\Big).
\end{align*}
Using \tcb{${\mu}_1(c)=\left(m\cdot\max\{c^{-\sigma_1},c^{-\sigma_2},c^{-\frac{1}{\sigma_2}}\}\right)^{-1}$} we obtain:
\begin{align*}
|I_1|\le c(|x|^{2-\xi_1}+\tcb{2}|x|^{2-\xi_2})+\frac{1}{{\mu}_1(c)}\left(|y|^{2-\xi_1}+\tcb{2}|y|^{2-\xi_2}\right).
\end{align*}
Therefore, using the definition of $\tilde{V}$, and by \eqref{ex1ylf}, we have
\begin{align*}
|I_1|&\le \tcb{2}c\tilde{V}(x)+\frac{1}{{\mu}_1(c)}\left(c_1^{\frac12\xi_1-1}\tilde{W}^2(y)+2c_1^{\frac12\xi_2-1}\tilde{W}^2(y)\right)\\
%&\leq 2c\tilde{V}^2(x)+\frac{1}{\mu_1(c)}c_1^{\frac12\xi_1-1}\tilde{W}^2(y)+\frac{2}{\mu_1(c)}c_1^{\frac12\xi_2-1}\tilde{W}^2(y)\\
&= \tcb{2}c\tilde{V}^2(x)+\frac{(c_1^{\frac12\xi_1-1}+\tcb{2}c_1^{\frac12\xi_2-1})}{\mu_1(c)}\tilde{W}(y)^2.
\end{align*}
%
% Next, we define the function $\mu_1:(0,\infty)\to(0, \infty)$ given by
% %
% \begin{equation}
%     \mu_1(s)=m(s^{-\sigma_1}+ s^{-\sigma_2}+s^{-\frac{1}{\sigma_2}}).
% \end{equation} 
% %
% where the constant $m>0$ is given by
% %
% It follows that:
% %
% \begin{align*}
%     |I_1|&\le c(|x|^{2-\xi_1}+2|x|^{2-\xi_2})+\delta_1^{1+\sigma_1}{c}^{-\sigma_1}|y|^{2-\xi_1}\\&~~~+\delta_2^{1+\sigma_2}{c}^{-\sigma_2}|y|^{2-\xi_2}+\delta_3^{1+\frac{1}{\sigma_2}}{c}^{-\frac{1}{\sigma_2}}|y|^{2-\xi_2}
%     \\&\le 2c\tilde{V}^2(x)+\mu_1({c})\tilde{W}^2(y)
% \end{align*}
%
for all $c>0$. Picking ${c}$ such that $0<c<\frac{1}{2L}\min\left\{{\kappa^{2-\xi_1}}, {\kappa^{2-\xi_2}}\right\}$ establishes \eqref{intu1} and \eqref{vcond}. To verify \eqref{intu2}, we have
\begin{align}\label{i2ex}
    I_2&=\nabla W(y)^\top\bj_{h}(x+\vt)\fx(P_\theta(x,y)).
\end{align}
Let $h^*:=\sup_{x}|\bj_{h}(x)|$, which exists and is finite since $h$ is $\mathcal{C}^1$ and globally Lipschitz. 
Moreover, we make note of the following fact:
\begin{small}
\begin{align*}
    |P_\theta(x,y)|\le |P_\theta(x,y)-P_\theta(x,0)|+|P_\theta(x,0)|\le \ell |y|+L|x|.
\end{align*}
\end{small}
With this in mind, 
We can proceed from \eqref{i2ex} as follows:
\begin{align*}
    |I_2|&\le \eta h^* |y|\left((\ell |y|+L|x|)^{1-\xi_1}+(\ell |y|+L|x|)^{1-\xi_2}\right)\\
    &\le \eta h^*|y|\big(\ell^{1-\xi_1}|y|^{1-\xi_1}+L^{1-\xi_1}|x|^{1-\xi_1}\\&~~~+2^{-\xi_2}(\ell^{1-\xi_2}|y|^{1-\xi_2}+L^{1-\xi_2}|x|^{1-\xi_2})\big)\\
    &\le \tilde{L}\left(|x|^{2-\xi_1}+|x|^{2-\xi_2}+|y|^{2-\xi_1}+|y|^{2-\xi_2}\right)\\
    &\le \tilde{L}\left(\tilde{V}^2(x)+2\tilde{W}^2(y)\right),
\end{align*}
where \tcb{$\tilde{L}:=\eta h^*\max\{\ell^{1-\xi_1}+ L^{1-\xi_1}, 2^{-\xi_2}\left(\ell^{1-\xi_2}+ L^{1-\xi_2}\right)\}$}. By applying Corollary \ref{cor_unif} we obtain the result.
\end{proof}
\vspace{0.1cm}

To illustrate Theorem \ref{exproof} via a numerical example, we simulate system \eqref{exint} {with the plant dynamics given by
$\dot{z}=-\mathcal{F}_{\frac25, -\frac27}(z-2\hx)$, which has the quasi-steady state $h(\hx)=2\hx$}. The cost function $\phi_\theta$ takes the quadratic form $\phi_\theta(\hx,z)=\frac12 \hx^\top Q_\theta z+b_\theta^\top z$, where $Q_\theta$ and $b_\theta$ are given by
\begin{equation*}
    Q_\theta:=\begin{bmatrix}
        3+d_1{(t)} & 2\\ 2& 5+d_2{(t)}
    \end{bmatrix},\quad b_\theta:=\begin{bmatrix}
        2+d_3{(t)}\\ 1
    \end{bmatrix}.
\end{equation*}
The parameters $d_i$ are given by $d_1(t)=0.8\sin(2.2\varepsilon\varepsilon_0 t)$, $d_2(t)=1.8\sin(1.7\varepsilon\varepsilon_0 t)$, $d_3(t)=0.66\sin(1.9\varepsilon\varepsilon_0 t)$. These signals can be generated by a system of the form \eqref{dtheta} by setting $\theta(t)\in\mathbb{R}^6$, $d_i(t)=\theta_{2i}(t)$, and $\Pi(\theta)=\mathcal{R}\theta$, where $\mathcal{R}\in\mathbb{R}^{6\times 6}$ is a block diagonal matrix with rotation matrices on the diagonal. Moreover, in accordance with Proposition \ref{exproof} we set $\xi_1=\frac13, \xi_2=-\frac15$. It can be verified that $Q_\theta\succ 0$ for all $t\ge 0$, and the optimizer of $\Phi_\theta$ is given by $\vt=-Q_\theta^{-1}b_\theta$. {We interconnect \eqref{exintb} with the plant dynamics, where $\hat{P}_\theta(\hx, z)=\frac12 Q_\theta z+Q_\theta \hx+2b_\theta$.} The trajectories of the system are shown in Figure \ref{ex2plot}, with $\varepsilon=0.05$ and different values of $\varepsilon_0$. As observed in the plot, the state $z$ converges in fixed-time to a neighborhood of the time-varying optimizer, whose size shrinks as $\varepsilon_0\to0^+$.
\section{Conclusion}\label{sec_conc}
We establish sufficient Lyapunov conditions for the study of FxT ISS properties in singularly perturbed systems. The results were applied to two illustrative examples: a particular nonsmooth second-order interconnection of systems, and a general fixed-time feedback optimization problem with time-varying cost functions, which has not been addressed before using fixed-time stability tools. Our method of verifying the interconnection conditions establishes an efficient paradigm for applying our results to other classes of algorithms and feedback schemes that exhibit multiple time scales. Future research directions include applying our results to a broader range of systems, including systems with more than two time scales. {Moreover, it is also of interest to identify a more general class of systems and quasi-steady state mappings for which our interconnection conditions hold, including characterizations based on homogeneity. %We hypothesize that our interconnection conditions hold for a general class of systems satisfying certain homogeneity (in the bi-limit) properties, but this remains an interesting and challenging open question to be addressed in future work.
}
\bibliographystyle{IEEEtran}
\bibliography{autosam.bib,Biblio.bib, noncoop}

@article{Fixed_timeTAC,
	Author = {A. Polyakov},
	Journal = {IEEE Transactions on Automatic and Control},
	Number = {8},
	Pages = {2106-2110},
	Title = {Nonlinear Feedback Design for Fixed-Time Stabilization of Linear Control Systems},
	Volume = {57},
	Year = {2012}}

@article{fixed_time,
	Author = {K. Garg and D. Panagou},
	Journal = {IEEE Transactions on Automatic Control},
	Title = {Fixed-time Stable Gradient-flow Schemes: Applications to Continuous-time Optimization},
	Year = {2018}}

@inbook{T,
	Author = {Mi{\c{e}}kisz, J.},
	Chapter = {Evolutionary Game Theory and Population Dynamics},
	Editor = {Capasso, V. and Lachowicz, M},
	Pages = {269--316},
	Publisher = {Springer},
	Title = {Lecture Notes in Mathematics. Multiscale Problems in the Life Sciences},
	Volume = {1940},
	Year = {2008}}

@book{khalil,
	Address = {Upper Saddle River, NJ},
	Author = {H. K. Khalil},
	Publisher = {Prentice Hall},
	Title = {Nonlinear Systems},
	Year = {2002}}

@article{TeelNesicTAC,
	Author = {A. R. Teel and L. Moreau and D. Nesic},
	Journal = {IEEE Transactions on Automatic Control},
	Pages = {1526-1544},
	Title = {A unified framework for input-to-state stability in systems with two time scales},
	Year = {2003}}

@article{nonsmoothesc,
	Author = {J. I. Poveda and M. Krsti\'{c}},
	Journal = {IEEE Transactions on Automatic Control},
	Number = {12},
	Pages = {6156-6163},
	Title = {Nonsmooth Extremum Seeking Control with User-Prescribed Fixed-Time Convergence},
	Volume = {66},
	Year = {2021}}

@article{SONTAG1995351,
title = {On characterizations of the input-to-state stability property},
journal = {Systems \& Control Letters},
volume = {24},
number = {5},
pages = {351-359},
year = {1995},
issn = {0167-6911},
doi = {https://doi.org/10.1016/0167-6911(94)00050-6},
author = {Eduardo D. Sontag and Yuan Wang}
}

@ARTICLE{1103586,
  author={Saberi, A. and Khalil, H.},
  journal={IEEE Transactions on Automatic Control}, 
  title={Quadratic-type {L}yapunov functions for singularly perturbed systems}, 
  year={1984},
  volume={29},
  number={6},
  pages={542-550},
  doi={10.1109/TAC.1984.1103586}}

@article{LOPEZRAMIREZ2020104775,
title = {Finite-time and fixed-time input-to-state stability: Explicit and implicit approaches},
journal = {Systems \& Control Letters},
volume = {144},
pages = {104775},
year = {2020},
issn = {0167-6911},
doi = {https://doi.org/10.1016/j.sysconle.2020.104775},
author = {Francisco Lopez-Ramirez and Denis Efimov and Andrey Polyakov and Wilfrid Perruquetti}
}

@ARTICLE{544001,
  author={Christofides, P.D. and Teel, A.R.},
  journal={IEEE Transactions on Automatic Control}, 
  title={Singular perturbations and input-to-state stability}, 
  year={1996},
  volume={41},
  number={11},
  pages={1645-1650},
  doi={10.1109/9.544001}}

@INPROCEEDINGS{261484,
  author={Van Breusegem, V. and Bastin, G.},
  booktitle={[1991] Proceedings of the 30th IEEE Conference on Decision and Control}, 
  title={Reduced order dynamical modelling of reaction systems: A singular perturbation approach}, 
  year={1991},
  volume={},
  number={},
  pages={1049-1054 vol.2},
  doi={10.1109/CDC.1991.261484}}

@ARTICLE{5467283,
  author={Moslehi, Khosrow and Kumar, Ranjit},
  journal={IEEE Transactions on Smart Grid}, 
  title={A Reliability Perspective of the Smart Grid}, 
  year={2010},
  volume={1},
  number={1},
  pages={57-64},
  doi={10.1109/TSG.2010.2046346}}

@book{del2015biomolecular,
  title={Biomolecular feedback systems},
  author={Del Vecchio, Domitilla and Murray, Richard M},
  year={2015},
  publisher={Princeton University Press Princeton, NJ}
}

@ARTICLE{9309064,
  author={Subotić, Irina and Groß, Dominic and Colombino, Marcello and Dörfler, Florian},
  journal={IEEE Transactions on Automatic Control}, 
  title={A {L}yapunov Framework for Nested Dynamical Systems on Multiple Time Scales With Application to Converter-Based Power Systems}, 
  year={2021},
  volume={66},
  number={12},
  pages={5909-5924},
  doi={10.1109/TAC.2020.3047368}}

@INPROCEEDINGS{9029355,
  author={Grunberg, Theodore W. and Del Vecchio, Domitilla},
  booktitle={2019 IEEE 58th Conference on Decision and Control (CDC)}, 
  title={Time-scale separation based design of biomolecular feedback controllers}, 
  year={2019},
  volume={},
  number={},
  pages={6616-6621},
  doi={10.1109/CDC40024.2019.9029355}}

@ARTICLE{28018,
  author={Sontag, E.D.},
  journal={IEEE Transactions on Automatic Control}, 
  title={Smooth stabilization implies coprime factorization}, 
  year={1989},
  volume={34},
  number={4},
  pages={435-443},
  doi={10.1109/9.28018}}

@INPROCEEDINGS{10644358,
  author={Tang, Michael and Krstic, Miroslav and Poveda, Jorge I.},
  booktitle={2024 American Control Conference (ACC)}, 
  title={On Fixed-Time Stability for a Class of Singularly Perturbed Systems Using Composite {L}yapunov Functions}, 
  year={2024},
  volume={},
  number={},
  pages={4783-4788},
  doi={10.23919/ACC60939.2024.10644358}}

@INPROCEEDINGS{9992641,
  author={Mendoza-Avila, Jesus and Efimov, Denis and Fridman, Leonid and Moreno, Jaime A.},
  booktitle={2022 IEEE 61st Conference on Decision and Control (CDC)}, 
  title={An Analysis of Convergence Properties of Finite-Time Homogeneous Controllers Through Its Implementation in a Flexible-Joint Robot}, 
  year={2022},
  volume={},
  number={},
  pages={5789-5794},
  doi={10.1109/CDC51059.2022.9992641}}

@ARTICLE{9540998,
  author={Bianchin, Gianluca and Cortés, Jorge and Poveda, Jorge I. and Dall’Anese, Emiliano},
  journal={IEEE Transactions on Control of Network Systems}, 
  title={Time-Varying Optimization of {L}{T}{I} Systems Via Projected Primal-Dual Gradient Flows}, 
  year={2022},
  volume={9},
  number={1},
  pages={474-486},
  doi={10.1109/TCNS.2021.3112762}}

@ARTICLE{10189107,
  author={Bastianello, Nicola and Carli, Ruggero and Zampieri, Sandro},
  journal={IEEE Transactions on Automatic Control}, 
  title={Internal Model-Based Online Optimization}, 
  year={2024},
  volume={69},
  number={1},
  pages={689-696},
  doi={10.1109/TAC.2023.3297504}}

@INPROCEEDINGS{10886634,
  author={Zimenko, K. and Efimov, D. and Polyakov, A. and Ping, X.},
  booktitle={2024 IEEE 63rd Conference on Decision and Control (CDC)}, 
  title={On Small-Gain Theorem for Interconnected Finite/Fixed-Time Input-to-State Stable Systems*}, 
  year={2024},
  volume={},
  number={},
  pages={4368-4372},
  doi={10.1109/CDC56724.2024.10886634}}

@ARTICLE{1101342,
  author={Chow, J. and Kokotovic, P.},
  journal={IEEE Transactions on Automatic Control}, 
  title={A decomposition of near-optimum regulators for systems with slow and fast modes}, 
  year={1976},
  volume={21},
  number={5},
  pages={701-705},
  doi={10.1109/TAC.1976.1101342}}

@ARTICLE{1104064,
  author={Saberi, A. and Khalil, H.},
  journal={IEEE Transactions on Automatic Control}, 
  title={Stabilization and regulation of nonlinear singularly perturbed systems--Composite control}, 
  year={1985},
  volume={30},
  number={8},
  pages={739-747},
  doi={10.1109/TAC.1985.1104064}}

@ARTICLE{10704051,
  author={Zhou, Yu and Polyakov, Andrey and Zheng, Gang},
  journal={IEEE Transactions on Automatic Control}, 
  title={Finite/Fixed-Time Stabilization of Linear Systems With State Quantization}, 
  year={2025},
  volume={70},
  number={3},
  pages={1921-1928},
  doi={10.1109/TAC.2024.3473619}}

@article{POLYAKOV2015332,
title = {Finite-time and fixed-time stabilization: Implicit {L}yapunov function approach},
journal = {Automatica},
volume = {51},
pages = {332-340},
year = {2015},
issn = {0005-1098},
doi = {https://doi.org/10.1016/j.automatica.2014.10.082},
author = {Andrey Polyakov and Denis Efimov and Wilfrid Perruquetti}
}

@article{naidu2001singular,
  title={Singular perturbations and time scales in guidance and control of aerospace systems: A survey},
  author={Naidu, D Subbaram and Calise, Anthony J},
  journal={Journal of Guidance, Control, and Dynamics},
  volume={24},
  number={6},
  pages={1057--1078},
  year={2001}
}

@book{KokotovicSPBook,
	author = {P. Kokotovi\'c and H. K. Khalil and J. O'Reilly},
	publisher = {Academic Press},
	title = {{S}ingular {P}erturbation {M}ethods in {C}ontrol: Analysis and Design},
	year = {1986}}

@book{narang2014nonlinear,
  title={Nonlinear Time Scale Systems in Standard and Nonstandard Forms},
  author={Narang-Siddarth, Anshu and Valasek, John},
  year={2014},
  publisher={SIAM}
}

@article{andrieu2008homogeneous,
  title={Homogeneous approximation, recursive observer design, and output feedback},
  author={Andrieu, Vincent and Praly, Laurent and Astolfi, Alessandro},
  journal={SIAM Journal on control and optimization},
  volume={47},
  number={4},
  pages={1814--1850},
  year={2008},
  publisher={SIAM}
}

@article{mendoza2023stability,
  title={On Stability of Homogeneous Systems in Presence of Parasitic Dynamics},
  author={Mendoza-Avila, Jes{\'u}s and Efimov, Denis and Fridman, Leonid and Moreno, Jaime A},
  journal={IEEE Transactions on Automatic Control},
  year={2023},
  publisher={IEEE}
}

@article{lei2022event,
  title={Event-Triggered Fixed-Time Stabilization of Two Time Scales Linear Systems},
  author={Lei, Yan and Wang, Yan-Wu and Mor{\u{a}}rescu, Irinel-Constantin and Postoyan, Romain},
  journal={IEEE Transactions on Automatic Control},
  volume={68},
  number={3},
  pages={1722--1729},
  year={2022},
  publisher={IEEE}
}

@article{polyakov2023finite,
  title={Finite-and fixed-time nonovershooting stabilizers and safety filters by homogeneous feedback},
  author={Polyakov, Andrey and Krstic, Miroslav},
  journal={IEEE Transactions on Automatic Control},
  year={2023},
  publisher={IEEE}
}

@article{poveda2022fixed,
  title={Fixed-time {N}ash equilibrium seeking in time-varying networks},
  author={Poveda, Jorge I and Krsti{\'c}, Miroslav and Ba{\c{s}}ar, Tamer},
  journal={IEEE Transactions on Automatic Control},
  volume={68},
  number={4},
  pages={1954--1969},
  year={2022},
  publisher={IEEE}
}

@article{vasil1978singular,
  title={Singular perturbations and some optimal control problems},
  author={Vasil’eva, AB and Dmitriev, MG},
  journal={IFAC Proceedings},
  volume={11},
  number={1},
  pages={963--966},
  year={1978},
  publisher={Elsevier}
}

@article{colombino2019online,
  title={Online optimization as a feedback controller: Stability and tracking},
  author={Colombino, Marcello and Dall’Anese, Emiliano and Bernstein, Andrey},
  journal={IEEE Transactions on Control of Network Systems},
  volume={7},
  number={1},
  pages={422--432},
  year={2019},
  publisher={IEEE}
}

@article{hauswirth2020timescale,
  title={Timescale separation in autonomous optimization},
  author={Hauswirth, Adrian and Bolognani, Saverio and Hug, Gabriela and D{\"o}rfler, Florian},
  journal={IEEE Transactions on Automatic Control},
  volume={66},
  number={2},
  pages={611--624},
  year={2020},
  publisher={IEEE}
}

@article{bianchin2022online,
  title={Online optimization of switched {L}{T}{I} systems using continuous-time and hybrid accelerated gradient flows},
  author={Bianchin, Gianluca and Poveda, Jorge I and Dall’Anese, Emiliano},
  journal={Automatica},
  volume={146},
  pages={110579},
  year={2022},
  publisher={Elsevier}
}

\section{Appendix}

We recall that Young's inequality states that $ab\leq \frac{a^p}{p}+\frac{b^q}{q}$ for all $a,b\geq0$, $p,q>1$ such that $\frac{1}{p}+\frac{1}{q}=1$. Also, we recall that, for a real-valued convex function $f$, Jensen's inequality states that:
\begin{equation}\label{jenseninequality}
f\left(\sum_{i=1}^nw_is_i\right)\leq \sum_{i=1}^nw_i f(s_i),~~\forall~s_i\in\mathbb{R},~
\end{equation}
where $\sum_{i=1}^n w_i=1$ and $w_i>0$ for all $i\in\{1,2,\ldots, n\}$. If $f$ is concave, then \eqref{jenseninequality} holds with the inequality reversed. As a consequence of \eqref{jenseninequality}, we obtain the following result
\begin{lemma}\label{jensenlemma}
    Let $s_i\ge 0$ for each $i\in \{1,2,...,n\}$. If $p\in (0,1]$, then $(\sum_{i=1}^n s_i)^p\le \sum_{i=1}^n s_i^p$. If $p>1$, then $(\sum_{i=1}^n s_i)^p\le n^{p-1}\sum_{i=1}^n s_i^p$.\QEDB
\end{lemma}

\vspace{0.1cm}
We present some additionally auxiliary lemmas that are utilized throughout our proofs:
\begin{lemma}\label{lem_sandw}
    Given $\underline{p}\le p\le \overline{p}$, the following holds
    \begin{equation*}
        x^p\le x^{\underline{p}}+x^{\overline{p}},
    \end{equation*}
    for all $x\ge 0$.\QEDB
\end{lemma}
\begin{proof}
    For $x\in [0,1]$ we have $x^p\le x^{\underline{p}}$, and for $x\ge 1$ we have $x^p\le x^{\overline{p}}$. We combine both cases to establish the result.
\end{proof}
\vspace{0.1cm}
\begin{lemma}\label{lemma_ex0}
    Given $\alpha>0$, the following holds:
    \begin{equation*}
        x\sg{x+u}^\alpha\ge 2^{-\alpha}|x|^{\alpha+1}
    \end{equation*}
    for all $x,u\in\re$ that satisfy $|x|>2|u|$.
\end{lemma}
\begin{proof}
    First assume $x<0$, which implies $|u|<-\frac12 x $. Since $\sg{x}^\alpha$ is strictly increasing in $x$, we have $\sg{x+u}^\alpha<\sg{x-\frac12 x}^\alpha=2^{-\alpha}\sg{x}.$
    Multiplying both sides by $x$ yields the result. Now suppose $x>0$, which implies $-|u|>-\frac12 x$. Then we have $\sg{x+u}^\alpha\ge \sg{x-|u|}^\alpha>\sg{x-\frac12 x}^\alpha=2^{-\alpha}\sg{x}^\alpha.$ We again multiply both sides by $x$ to obtain the result.
\end{proof}

As a direct consequence of Lemma \ref{amgm}, we have the following additional result
\begin{lemma}\label{amgm_abs}
    Given $x, y\ge 0$ and $p_1, p_2>0$, the following inequality holds:
    \begin{equation*}
        \delta|x|^{p_1}|y|^{p_2}\le {c}|x|^{p_1+p_2}+\delta^{1+\frac{p_1}{p_2}}{c^{-\frac{p_1}{p_2}}}|y|^{p_1+p_2}
    \end{equation*}
    for all $\delta>0$ and {$c>0$}.\QEDB
\end{lemma}
\begin{proof}
Let $\tilde{c}=\frac{c}{\delta}>0$.
    By Lemma \ref{amgm}, we have
    \begin{align*}
        \delta|x|^{p_1}|y|^{p_2}&\le \delta\left(\tilde{c}|x|^{p_1+p_2}+\tilde{c}^{-\frac{p_1}{p_2}}|y|^{p_1+p_2}\right)\\
        &={c}|x|^{p_1+p_2}+\delta^{1+\frac{p_1}{p_2}}{c^{-\frac{p_1}{p_2}}}|y|^{p_1+p_2}
    \end{align*}
    which establishes the result.
\end{proof}
\subsection{Proof of Lemma \ref{amgm}}
\begin{proof} Let $p_1,p_2>0$ and $x,y>0$  be given (note that the case $x=0$ or $y=0$ follows trivially). For any $c>0$, define $\tilde{c}:=\left(\frac{(p_1+p_2)}{p_1}c\right)^{\frac{p_1}{p_1+p_2}}$. Using $a:=\tilde{c}|x|^{p_1}$, $b:=\frac{1}{\tilde{c}}|y|^{p_2}$ and $p:=\frac{p_1+p_2}{p_1}$ (which implies $q=\frac{p_1+p_2}{p_2}$) we directly obtain:
\begin{align*}
&|x|^{p_1}|y|^{p_2}\leq \frac{p_1}{p_1+p_2}\left(\tilde{c}|x|^{p_1}\right)^{\frac{p_1+p_2}{p_1}}+\frac{p_2}{p_1+p_2}\left(\frac{1}{\tilde{c}}|y|^{p_2}\right)^{\frac{p_1+p_2}{p_2}}\\
&~~~~=\frac{p_1}{p_1+p_2}\tilde{c}^{\frac{p_1+p_2}{p_1}}|x|^{p_1+p_2}+\frac{p_2}{p_1+p_2}\left(\frac{1}{\tilde{c}}\right)^{\frac{p_1+p_2}{p_2}}|y|^{p_1+p_2}.
\end{align*}
Since $c=\frac{p_1}{p_1+p_2}\tilde{c}^{\frac{p_1+p_2}{p_1}}$ and $\tilde{c}^{\frac{p_1+p_2}{p_2}}=\left(\frac{(p_1+p_2)}{p_1}c\right)^{\frac{p_1}{p_2}}$ we obtain
\begin{align*}
|x|^{p_1}|y|^{p_2}&\leq c|x|^{p_1+p_2}+\frac{p_2}{p_1+p_2}\left(\frac{1}{\left(\frac{(p_1+p_2)}{p_1}c\right)^{\frac{p_1}{p_2}}}\right)|y|^{p_1+p_2}\\
&=c|x|^{p_1+p_2}+\frac{p_2}{p_1+p_2}\left(\frac{p_1}{(p_1+p_2)c}\right)^{\frac{p_1}{p_2}}|y|^{p_1+p_2}
\end{align*}
Since $\frac{p_2}{p_1+p_2}<1$ and $(\frac{p_1}{p_1+p_2})^{\frac{p_1}{p_2}}<1$, we finally obtain
    \begin{align*}
        |x|^{p_1}|y|^{p_2}\le c |x|^{p_1+p_2}+c^{-\frac{p_1}{p_2}}|y|^{p_1+p_2},
    \end{align*}
    which establishes the result.
\end{proof}
\subsection{Proof of Lemma \ref{bounds}}
\begin{proof}
    First, we establish \eqref{bdxi1}. Since $|s|^2=s^\top s$ for any vector $s\in\mathbb{R}^n$,
    squaring both sides of \eqref{bdxi1} we see that it suffices to show:

\vspace{-0.2cm}
\begin{small}
\begin{align}\label{t1gen}
    |x|^{2-2{\xi_1}}+|y|^{2-2{\xi_1}}-\frac{2 x^\top y}{|x|^{{\xi_1}} |y|^{{\xi_1}}}&\le 4^{{\xi_1}} \Big(|x|^2+|y|^2-2x^\top y\Big)^{1-{\xi_1}}.
\end{align}
\end{small}

\vspace{-0.2cm}\noindent 
To show this inequality we consider three cases: $x^\top y> 0$, $x^\top y<0$, and $x^\top y=0$. First, suppose $x^\top y>0$. It suffices to show the stronger inequality
\begin{small}
\begin{align}\label{t1case1}
    |x|^{2-2{\xi_1}}+|y|^{2-2{\xi_1}}-\frac{2 x^\top y}{|x|^{{\xi_1}} |y|^{{\xi_1}}}&\le 2^{{\xi_1}} \Big(|x|^2+|y|^2-2x^\top y\Big)^{1-{\xi_1}}.
\end{align}
\end{small}

\vspace{-0.3cm}\noindent 
We assume $|x|^{2-2{\xi_1}}+|y|^{2-2{\xi_1}}\neq 0$, because otherwise we would have $x=y=0$, from which the result trivially follows. 

We note that inequality \eqref{t1case1} is homogeneous. Indeed, if we fix some $\delta>0$ and replace $x$ and $y$ in \eqref{t1case1} with $\delta x$ and $\delta y$ respectively, we obtain
\begin{align}
    |\delta x|^{2-2{\xi_1}}+|\delta y|^{2-2{\xi_1}}-\frac{2 (\delta x)^\top (\delta y)}{|\delta x|^{{\xi_1}} |\delta y|^{{\xi_1}}}&\le 2^{{\xi_1}} \Big(|\delta x|^2+|\delta y|^2\notag\\
    &~~-2(\delta x)^\top (\delta y)\Big)^{1-{\xi_1}}\label{rescaledinequality}.
\end{align}
Factorizing $\delta$ at the left and right-hand sides, we obtain 
\begin{align*}
    \delta^{2-2\xi_1}|x|^{2-2{\xi_1}}&+\delta^{2-2\xi_1}|y|^{2-2{\xi_1}}-\frac{\delta^{2-2\xi_1} 2x^\top y}{|x|^{{\xi_1}} | y|^{{\xi_1}}}\\
    &\le 2^{{\xi_1}}\delta^{2-2\xi_1}\Big(| x|^2+|y|^2-2x^\top y\Big)^{1-{\xi_1}},
\end{align*}
which is equivalent to \eqref{t1case1} after dividing both sides by $\delta^{2-2\xi_1}$.

For any $x,y\in\mathbb{R}^n\setminus\{0\}$ and $\xi_1\in(0,1)$, we have $|x|^{2-2{\xi_1}}+|y|^{2-2{\xi_1}}=p$ for some $p>0$.  Dividing both sides by $p$, we obtain
\begin{equation}\label{unitsum}
|\delta x|^{2-2{\xi_1}}+|\delta y|^{2-2{\xi_1}}=1.
\end{equation}
where $\delta:=\frac{1}{p^{\frac{1}{2-2\xi_1}}}$. Thus, it suffices to show \eqref{rescaledinequality} using \eqref{unitsum} with $\tilde{x}=\delta x$ and $\tilde{y}=\delta y$. Therefore, without loss of generality and to avoid introducing new notation, we can assume that $|x|^{2-2{\xi_1}}+|y|^{2-2{\xi_1}}=1$.  Under this condition, we can state the following:
 \begin{align}
     |x|^{{\xi_1}} |y|^{{\xi_1}}&\le \left(\frac{|x|^{2-2{\xi_1}}+|y|^{2-2{\xi_1}}}{2}\right)^\frac{{\xi_1}}{1-{\xi_1}}=2^{\frac{{\xi_1}}{{\xi_1}-1}},\label{keystep1lemma2}
 \end{align}
where the inequality follows directly by Young's inequality with $p=q=2$, $a=|x|^{1-\xi_1}$ and $b=|y|^{1-\xi_1}$.

Similarly, using \eqref{jenseninequality} with $w_i=\frac{1}{2}$, $s_1=|x|^{2-2\xi_1}$, $s_2=|y|^{2-2\xi_1}$, and $f(\cdot)=(\cdot)^\frac{1}{1-\xi_1}$ (which is convex since $\xi_1\in(0,1)$) we have the inequality
\begin{equation}\label{auxjense}
\frac{1}{2}|x|^2+\frac{1}{2}|y|^2\geq \left(\frac{1}{2}|x|^{2-2\xi_1}+\frac{1}{2}|y|^{2-2\xi}\right)^{\frac{1}{1-\xi_1}}.
\end{equation}
Using $|x|^{2-2{\xi_1}}+|y|^{2-2{\xi_1}}=1$ and multiplying both sides of \eqref{auxjense} by $2$, we obtain
\begin{equation}\label{keystepjense}
     |x|^2+|y|^2\ge 2^{\frac{{\xi_1}}{{\xi_1}-1}}.
\end{equation}
Using \eqref{keystep1lemma2} and \eqref{keystepjense}, it suffices to show
\begin{equation*}
         1-\frac{2x^\top y}{2^\frac{{\xi_1}}{{\xi_1}-1}}\le 2^{{\xi_1}} \left(2^{\frac{{\xi_1}}{{\xi_1}-1}}-2x^\top y\right)^{1-{\xi_1}},
\end{equation*}
which is equivalent to
\begin{equation}\label{t1gendone}
         \left(2^{\frac{{\xi_1}}{{\xi_1}-1}}-2x^\top y\right)^{{\xi_1}}\le 2^{\frac{{\xi_1}^2}{{\xi_1}-1}}.
\end{equation}
Notice that since $0<x^\top y\le |x||y|\le 2^\frac{1}{{\xi_1}-1}$ (which follows by Cauchy-Schwartz inequality and \eqref{keystep1lemma2}), we have $0\le 2^{\frac{{\xi_1}}{{\xi_1}-1}}-2x^\top y<2^{\frac{{\xi_1}}{{\xi_1}-1}}$, which implies \eqref{t1gendone}. 

Now, we consider the case where $x^\top y<0$. First, we can denote $|x|^{2-2{\xi_1}}+|y|^{2-2{\xi_1}}=p$, for some $p>0$. Let $\delta:=\frac{|x^\top y|}{|x||y|}$, which is well-defined because $x\neq0$ and $y\neq0$. By Cauchy-Schwartz we have that $\delta\in(0,1]$. It follows that \eqref{t1gen} becomes
 \begin{equation}\label{intermediateboundlemma2}
     p+2\delta|x|^{1-{\xi_1}}|y|^{1-{\xi_1}}\le 4^{{\xi_1}}\left(|x|^2+|y|^2+2\delta|x||y|\right)^{1-{\xi_1}}.
 \end{equation}
Now take $\tilde{w}_1, \tilde{w}_2>1$ satisfying $\frac{2}{\tilde{w}_1}+\frac{1}{\tilde{w}_2}=1$. Using the fact that $f(\cdot)=(\cdot)^{1-{\xi_1}}$ is concave on $\mathbb{R}_{\ge0}$ because $\xi_1\in(0,1)$, we obtain the following via the reverse Jensen's inequality:
\begin{align*}      
&4^{{\xi_1}}\left(\frac{1}{\tilde{w}_1}\tilde{w}_1|x|^2+\frac{1}{\tilde{w}_1}\tilde{w}_1|y|^2+\frac{1}{\tilde{w}_2}\tilde{w}_22\delta|x||y|\right)^{1-{\xi_1}}
\\
&\geq 4^{\xi_1}\Bigg(\frac{1}{\tilde{w}_1}\tilde{w}_1^{1-\xi_1}|x|^{2-2\xi_1} + \frac{1}{\tilde{w}_1}\tilde{w}_1^{1-\xi_1}|y|^{2-2\xi_1}\\
&~~~~+\frac{1}{\tilde{w}_2}(2\tilde{w}_2\delta)^{1-\xi_1}(|x||y|)^{1-\xi_1}\Bigg)\\
&=4^{{\xi_1}}\left(\tilde{w}_1^{-{\xi_1}}p+\tilde{w}_2^{-{\xi_1}}(2\delta)^{1-{\xi_1}}|x|^{1-{\xi_1}} |y|^{1-{\xi_1}}\right).
\end{align*}
Thus, it suffices to establish the following stronger bound
 \begin{align*}
     p+2\delta|x|^{1-{\xi_1}}|y|^{1-{\xi_1}}&\le 4^{{\xi_1}}\Big(\tilde{w}_1^{-{\xi_1}}p\\&~~~+\tilde{w}_2^{-{\xi_1}}(2\delta)^{1-{\xi_1}}|x|^{1-{\xi_1}} |y|^{1-{\xi_1}}\Big),
 \end{align*}
which holds if $4^{\xi_1}\tilde{w}_1^{-\xi_1}>1$ and $4^{\xi_1}\tilde{w}_2^{-\xi_1}(2\delta)^{1-\xi_1}>2\delta$. Thus, since 
\begin{equation}
\max\left(\tilde{w}_1^{{\xi_1}}, (2\tilde{w}_2)^{{\xi_1}}\right)\ge \max\left(\tilde{w}_1^{{\xi_1}}, \frac{\tilde{w}_2^{{\xi_1}}2\delta}{(2\delta)^{1-{\xi_1}}} \right),
\end{equation}
it suffices to have $4^{{\xi_1}}\ge \max\left(\tilde{w}_1^{{\xi_1}}, (2\tilde{w}_2)^{{\xi_1}}\right)$. Since $2\tilde{w}_2=\frac{2\tilde{w}_1}{\tilde{w}_1-2}$, the above condition holds with $\tilde{w}_1=4,~\tilde{w}_2=2$, which establishes the bound \eqref{intermediateboundlemma2}.

Finally, we consider the case $x^\top y=0$. Here, \eqref{t1gen} reduces to 
\begin{align}
    |x|^{2-2{\xi_1}}+|y|^{2-2{\xi_1}}&\le 4^{{\xi_1}} \Big(|x|^2+|y|^2\Big)^{1-{\xi_1}}.
\end{align}
which follows directly by 
\begin{align}
    |x|^{2-2{\xi_1}}+|y|^{2-2{\xi_1}}&\le 2^{{\xi_1}} \Big(|x|^2+|y|^2\Big)^{1-{\xi_1}},
\end{align}
which is a direct application of the reverse Jensen's inequality. Combining the above three cases establishes \eqref{bdxi1}. 
%We essentially want to pick $w_1, w_2$ such that $\max\left(w_1^{a_1}, (2w_2)^{a_1}\right)$ is minimiyed. 
%Since $2w_2=\frac{2w_1}{w_1-2}$ is strictly decreasing in $w_1$, we see that $\max\left(w_1^{{\xi_1}}, (2w_2)^{{\xi_1}}\right)$ is minimized at %$w_1$ satisfying $w_1>2$ and $w_1=2w_2$, which is at 
%$w_1=4$. 
%Therefore at $w_1=4, w_2=2$ we have $4^{{\xi_1}}=\max\left(w_1^{{\xi_1}}, (2w_2)^{{\xi_1}}\right)$ and we are done.

\vspace{0.2cm}
Next, we consider \eqref{bdxi2}. We can assume $x\neq 0$ and $y\neq 0$, since otherwise the statement holds for $K=1$. We first consider the case where $|x|\ge |y|>0$. We have
\begin{align*}
    \left\lvert \dfrac{x}{|x|^{{\xi_2}}}-\dfrac{y}{|y|^{{\xi_2}}} \right\rvert&\le 
|x-y||x|^{-{\xi_2}}+|y|\left\lvert|x|^{-{\xi_2}}-|y|^{-{\xi_2}}\right\rvert, 
\end{align*}
which is obtained by adding and subtracting $\frac{y}{|x|^{\xi_2}}$ to the terms in the left side of the inequality and using the triangle inequality. Since $||x|-|y||\le |x-y|$, it suffices to show
\begin{align}
    &|y|\left\lvert|x|^{-{\xi_2}}-|y|^{-{\xi_2}}\right\rvert\notag\\&\le K\left\lvert|x|-|y|\right\rvert^{-{\xi_2}+1}+(K-1)|x|^{-{\xi_2}}\left\lvert|x|-|y|\right\rvert.\label{step1}
\end{align}
Let $p:=\frac{|x|}{|y|}\ge 1$; Then, \eqref{step1} is equivalent to
\begin{equation}
    p^{-{\xi_2}}-1\le K(p-1)^{-{\xi_2}+1}+(K-1)(p-1)p^{-{\xi_2}}.\label{step2}
\end{equation}
We can define the function $f(p):=K(p-1)^{-{\xi_2}+1}+(K-1)(p-1)p^{-{\xi_2}}-(p^{-{\xi_2}}-1)$, and observe that \eqref{step2} holds if $f(p)\ge 0$ for all $p\ge 1$. To this end, note that $f(1)=0$, so it suffices to show $f'(p)\ge 0$ for $p\ge 1$. This is equivalent to establishing
\begin{equation*}
    K(p-1)^{-{\xi_2}}+(K-1)p^{-{\xi_2}}\ge \frac{-{\xi_2} K}{-{\xi_2}+1}p^{-{\xi_2}-1}.
\end{equation*}
Dividing both sides by $(K-1)(p^{-\xi_2-1})$, we obtain:
\begin{equation*}
    \frac{K}{K-1}\frac{(p-1)^{-\xi_2}}{p^{-\xi_2-1}}+p\ge \left(\frac{-{\xi_2}}{-{\xi_2}+1}\right)\left(\frac{K}{K-1}\right).
\end{equation*}
Since the first term in the left-hand side of the inequality is positive, it suffices to have
\begin{equation*}
    p\ge \left(\frac{-{\xi_2}}{-{\xi_2}+1}\right)\left(\frac{K}{K-1}\right).
\end{equation*}
for all $p\geq1$. In turn, under the definition of $K$, this condition holds if $K\ge -{\xi_2}+1$, which is true since $\max\left(1,-{\xi_2} 2^{-{\xi_2}-1}\right)\ge -{\xi_2}$ for ${\xi_2}<0$. 

\vspace{0.1cm}
Now, suppose $|y|>|x|>0$. We again proceed from \eqref{step1}, but now we have $0<p<1$, and \eqref{step1} is equivalent to:
 \begin{equation}
    1-p^{-{\xi_2}}\le K(1-p)^{-{\xi_2}+1}+(K-1)(1-p)p^{-{\xi_2}}.
\end{equation}
Then it suffices to show the following for all $p\in(0,1)$:
 \begin{equation}
    1-p^{-{\xi_2}}\le (K-1)(1-p)^{-{\xi_2}+1}+(K-1)(1-p)p^{-{\xi_2}}.
\end{equation}
which leads to
\begin{equation}\label{finalboundlemma}
    K-1\ge \frac{1-p^{-{\xi_2}}}{(1-p)^{-{\xi_2}+1}+(1-p)p^{-{\xi_2}}}:=g(p).
\end{equation}
We claim that $\eqref{finalboundlemma}$ holds under the definition of $K$, since  $\sup_{p\in(0,1)}g(p)\le -{\xi_2}2^{-{\xi_2}-1}$ for ${\xi_2}\le -1$ and $\sup_{p\in (0,1)}g(p)\le 1$ for ${\xi_2}\in (-1,0)$. These estimates can be obtained by rewriting $g$ as:
\begin{equation*}
    g(p)=\dfrac{\left(\dfrac{1-p^{-{\xi_2}}}{1-p}\right)}{(1-p)^{-{\xi_2}}+p^{-{\xi_2}}},
\end{equation*}
followed by maximizing the numerator on $p\in (0,1)$ and minimizing the denominator on $p\in (0,1)$.
\end{proof}

\end{document}